\documentclass[runningheads]{llncs}

\usepackage[T1]{fontenc}
\usepackage{graphicx}
\usepackage[utf8]{inputenc}
\usepackage{microtype}
\usepackage{amsmath,amssymb, color, tikz, mathtools}
\usepackage{tcolorbox}
\usepackage{mdframed}
\usepackage{color}
\usepackage{colortbl}
\usetikzlibrary{positioning}
\usetikzlibrary{shapes}
\usetikzlibrary{decorations.pathreplacing}
\usepackage{enumerate, enumitem} 

\usepackage{hyperref}
\makeatletter
\newcommand{\printfnsymbol}[1]{%
  \textsuperscript{\@fnsymbol{#1}}%
}
\makeatother
\usepackage{graphicx}
\graphicspath{{./image/}}

\hypersetup{colorlinks,linkcolor={blue},citecolor={blue},urlcolor={blue}}
\usepackage[ruled,vlined]{algorithm2e}
\usepackage{algpseudocode}

\newcommand{\N}{\mathbb{N}}

\newcommand{\zone}{\{0, 1\}}

\renewcommand{\deg}{\mathrm{deg}}

\newcommand{\supp}{\mathrm{supp}}

\newcommand{\cert}{\mathsf{C}}
\newcommand{\fcert}{\mathsf{FC}}

\newcommand{\bs}{\mathsf{bs}}
\newcommand{\fbs}{\mathsf{fbs}}
\newcommand{\s}{\mathsf{s}}
\newcommand{\ms}{\mathsf{ms}}
\newcommand{\hsc}{\mathsf{MCC}}

\newcommand{\spar}{\mathsf{spar}}

\newcommand{\AND}{\mathsf{AND}}

\newcommand{\OR}{\mathsf{OR}}

\renewcommand{\hat}{\widehat}
\renewcommand{\tilde}{\widetilde}

\newtheorem{open question}[theorem]{Open question}
\begin{document}

\title{Relations between monotone complexity measures based on decision tree complexity}

\author{Farzan Byramji \inst{1}$^,$ \thanks{Work done while FB and VJ were at Indian Institute of Technology, Kanpur} \and Vatsal Jha \inst{2}$^,$ \printfnsymbol{1} \and Chandrima Kayal \inst{3} \and Rajat Mittal \inst{4}}
\institute{ University of California San Diego \\\email{fbyramji@ucsd.edu} 
\and Purdue University\\ \email{jha36@purdue.edu} \and  Indian Statistical Institute, Kolkata \\\email{chandrimakayal2012@gmail.com} \and Indian Institute of Technology, Kanpur\\ \email{rmittal@cse.iitk.ac.in}}
\authorrunning{F.~Byramji, V.~Jha, C.~Kayal and  R.~Mittal}
%\authorrunning{N. Author}

%
\maketitle              % typeset the header of the contribution

\begin{abstract}

In a recent result, Knop, Lovett, McGuire and Yuan (STOC 2021) proved the log-rank conjecture for communication complexity, up to $\log n$ factor, for any Boolean function composed with $\AND$ function as the inner gadget. One of the main tools in this result was the relationship between monotone analogues of well-studied Boolean complexity measures like block sensitivity and certificate complexity. The relationship between the standard measures has been a long line of research, with a landmark result by Huang (Annals of Mathematics 2019), finally showing that sensitivity is polynomially related to all other standard measures.
%\rajat{Should we change HSC to MCC, to give it a feel of results on monotone versions?}
 
In this article, we study the monotone analogues of standard measures like block sensitivity ($\mathsf{mbs}(f)$), certificate complexity ($\mathsf{MCC}(f)$) and fractional block sensitivity ($\mathsf{fmbs}(f)$); and study the relationship between these measures given their connection with $\AND$-decision tree and sparsity of a Boolean function. We show the following results:
\begin{itemize}
    \item Given a Boolean function $f:\zone^n \rightarrow \zone$, the ratio $\frac{\mathsf{fmbs}(f^{l})}{\mathsf{mbs}(f^{l})}$ is bounded by a function of $n$ (and not $l$). A similar result was known for the corresponding standard measures (Tal, ITCS 2013). This result allows us to extend any upper bound by a \emph{well behaved} measure on monotone block sensitivity to monotone fractional block sensitivity. 
    \item The question of the best possible upper bound on monotone block sensitivity by the logarithm of sparsity is equivalent to the natural question of best upper bound by degree on sensitivity. One side of this relationship was used in the proof by Knop, Lovett, McGuire and Yuan (STOC 2021).
     
    \item For two natural classes of functions, symmetric and monotone, hitting set complexity ($\mathsf{MCC}$) is equal to monotone sensitivity. 
    %This implies that the deterministic communication complexity for any function $f$ in these classes is bounded by $\log^3 \spar(f) \log n$ as compared to $\log^5 \spar(f) \log n$ for a general function, where $\spar(f)$ is the sparsity of $f$.
    %\item Given a Boolean function $f:\zone^n \rightarrow \zone$, in terms of composition of block sensitivity $(\bs)$ it is known that $\bs(f \circ g) \leq \fbs(f).\bs(g)$. We have shown the similar results in monotone settings which is the following: $\mbs(f \circ g) \leq \fbs(f). \mbs(g)$.
   % \chandrima{while proving the ratio it has been also discussed how the monotone measures behave under composition. We may want to sale that also}
\end{itemize}

% For a Boolean function $f : \zone^n \to \zone$, understanding the relationship between various combinatorial measures has been a central area of research for many years. Recently Alexander Knop, ... introduced this complexity measures for studying log rank and lifting for AND function. In this regards exploring these newly introduced monotone measures (e.g. monotone sensitivity $(\mathsf{ms})$, monotone block sensitivity $(\mathsf{mbs})$, monotetne fractional block sensitivity $(\mathsf{fmbs})$, hitting set complexity $(\mathsf{HSC})$ etc) , their properties and how they behave have independent interest itself. In this paper we have studied: 
    
%     \begin{itemize}
%         \item the ratio of fractional monotone block sensitivity and monotone block sensitivity under the composition of same function and their applications. The corresponding result is already known in general set up that is for the measures fractional block sensitivity and block sensitivity. 
%         \item relations and separations between monotone measures and other complexity measures. 
%     \end{itemize} 
%     
\end{abstract}

% \input{cocoon24/mainlncs.tex/ratiotest}

% For the converse,
% \end{proof}

%\rajat{Please spell check the document.}
%\rajat{should autoref be used for references?}
%\rajat{we are stating same theorems with different numbers.}
%%%%%%%%%%%%%%%%%%%%%%%%%%%%%%%%%%%%%%%%%%%%%%%%%%%%%%%%%%%%%%%%%%
\section{Introduction}
%%%%%%%%%%%%%%%%%%%%%%%%%%%%%%%%%%%%%%%%%%%%%%%%%%%%%%%%%%%%%%%%%%

%\rajat{Functions separating monotone measures?}
%\rajat{Upper bounds by standard measures on monotone measures?}

%1.\textcolor{blue}{please be consistent about the notations!$\bs,\fbs,\dots$}\\
%2. \textcolor{red}{Boolean}\\

Decision tree complexity is one of the simplest complexity measure for a Boolean function where the complexity of an algorithm only takes into account the number of queries to the input. Various complexity measures based on decision tree complexity (like quantum query complexity, randomized query complexity, certificate complexity, sensitivity, block sensitivity and many more) have been introduced to study Boolean functions (functions from a subset of $\zone^n$ to $\zone$)~\cite{BW,ABK16,NS94}. Understanding the relations between these complexity measures of Boolean function has been a central area of research in computational complexity theory for at least 30 years. Refer to~\cite{BW} for an introduction to this area. 

Two such complexity measures $M_1$ and $M_2$ are said to be polynomially related if there exists constants $c_1$ and $c_2$ such that $M_1= O(M_2^{c_1})$ and $M_2 = O(M_1^{c_2})$.  Recently, Huang \cite{Huang} resolved a major open problem in this area known as the ``sensitivity conjecture'', showing the polynomial relationship between the two complexity measures sensitivity $(\s(f))$ and block sensitivity $(\bs(f))$ for a Boolean function $f$ (implying sensitivity is polynomially related to almost all other complexity measures too).

% For studying Boolean functions $f: \zone^n \to \zone$ there are various complexity measures like deterministic query complexity, randomized query complexity, certificate complexity, sensitivity, block sensitivity and many more. \rajat{put references here. } All of these measures are based on decision tree complexity. Now for any two different complexity measures $M_1$ and $M_2$ we say they are polynomially related if there exists constants $c_1$ and $c_2$ such that $M_1= O(M_2^{c_1})$ and $M_2 = O(M_1^{c_2})$. Understanding the relations between various complexity measures of Boolean function has been a central area of research in computational complexity theory for many years. Recently Huang \cite{Huang} resolved a major open problem in this area known as the ``sensitivity conjecture''. which asks about the polynomial relationship between two complexity measures sensitivity $(\s(f))$ and block sensitivity $(\bs(f))$. For an introduction to this area, we refer to the survey~\cite{BW} and for the present state of research in this area we refer to \cite{ABK+}, where a nice table have been compiled with the best known relationship between different measures of Boolean functions.

Once two complexity measures have been shown to be polynomially related, it is natural to ask if the relationships are tight or not. This means, if we can show $\forall~ f,~M_1(f) = O(M_2(f))^{\alpha}$, then is there an example that witnesses the same gap? In other words, does there exists a function $f$ for which $M_1(f) = \Omega(M_2(f))^{\alpha}$?
Figuring out tight relations between complexity measures based on decision trees has become the central goal of this research area. (\cite{ABK+} compiled an excellent table with the best-known relationships between these different measures.) 

Additionally, many new related complexity measures have been introduced in diverse areas, sometimes to understand these relations better~\cite{JainK10,BB20,CGLMS23}. Recently, monotone analogues of such combinatorial measures have been explored in \cite{KLMY21} for studying the celebrated \emph{log-rank conjecture} in communication complexity (for definitions of these monotone measures, see \autoref{preliminaries}).

 In particular, Knop et. al.~\cite{KLMY21} resolved the log rank conjecture (up to $\log(n)$ factor) for any Boolean function $f$ composed with $\AND$ function as the inner gadget. 
 For such functions, the rank is equal to the sparsity of the function (denoted by $\spar(f))$ in its polynomial representation (with range $\zone$). So the log-rank conjecture amounts to proving a polynomial upper bound of $\log(\spar(f))$ on the deterministic communication complexity of $(f\circ \wedge_{2})$ i.e. $D^{cc}(f\circ \wedge_{2})$. 
 As mentioned before, the proof provided in \cite{KLMY21} utilized monotone analogues of the standard combinatorial measures (block sensitivity, fractional block sensitivity etc.). 
 The reason for considering monotone measures was due to the observation that the deterministic communication complexity of such functions is related to their fractional monotone block sensitivity ($\mathsf{fmbs}$)
 $$D^{cc}(f\circ \wedge_{2})\leq \mathsf{fmbs}(f) \log(\spar(f)) \log(n).$$

% \rajat{Shouldn't this be HSC instead of fmbs, we can probably mention that HSC is same as MCC.}

% While the main problem is still open, various significant improvements have been made
% considering XOR as inner function. Recently, in \cite{KLMY21} resolved the log rank conjecture (upto $\log(n)$ factor) for any Boolean function composed with $\mathsf{AND}$ function as the inner gadget. The proof provided in \cite{KLMY21} utilised monotone analogues of some classical combinatorial measures of Boolean functions like sensitivity, block sensitivity, certificate complexity etc. The reason for considering monotone measures was due to the observation that a sparse Boolean function has small $D^{0-dt}(f)$ (see \cite{KLMY21}) i.e. : 
% $$D^{0-dt}(f)=O(\mathsf{fmbs}(f)\log(\spar(f)))=O(\log^{5}(\spar(f))).$$  

%For example, we extend some of existing results of \cite{Tal13} to the monotone settings. 

% This in turn showed the existence of a small $\mathsf{AND}$-decision tree for functions and hence resolved the log-rank conjecture for communication problems obtained from composing Boolean functions $f:\zone^{n}\rightarrow\zone$ with the two-bit $\mathsf{AND}$-gadget ($\wedge_{2}$), i.e.:
% $$D^{cc}(f\circ \wedge_{2})\leq 2D^{\wedge-dt}(f)=O(\log^{5}(\spar(f))\log(n)).$$

Fractional block sensitivity (and its relation with block sensitivity) has been studied before~\cite{Aaronson20,tal-kul,Tal13,GSS16}. It seems natural to look at the monotone analogues of fractional block sensitivity ($\mathsf{fmbs}$) and block sensitivity ($\mathsf{mbs}$), and see if they can used to upper bound $\mathsf{fmbs}$ with logarithm of sparsity. They precisely do this, and show that for every Boolean function $f:\zone^{n}\rightarrow\zone$: 

\begin{itemize}
\item $\mathsf{mbs}(f) = O(\log^{2}(\spar(f)))$,
\item $\mathsf{fmbs}(f) = O(\mathsf{mbs}^{2}(f))$.
\end{itemize}

This implies that a sparse Boolean function has small fractional monotone block sensitivity and in turn small deterministic communication complexity,
$$\mathsf{fmbs}(f) = O(\log^{4}(\spar(f))) \Rightarrow D^{cc}(f\circ \wedge_{2}) = O(\log^{5}(\spar(f)) \log n). $$ 

Here $\spar(f)$ denotes the sparsity of $f$ as a polynomial with range $\zone$.
% (the exponent can be improved by directly relating communicating complexity with fractional monotone block sensitivity).
% A natural question to ask is whether we can improve the $D^{\wedge-dt}(f)$ from $O\log^{5}(\spar(f))\log(n)$? This is in fact one of the main questions that we address in our current work. Our approaches are motivated by the following results from~\cite{KLMY21}, which have been used to prove $\hsc(f)=O(\log^{5}(\spar(f)))$:
% that a sparse Boolean function has small Hitting set complexity ($\hsc$) , they show  
% $$\hsc(f) = O(\log^{5}(\spar(f))).$$  

% \textcolor{red}{lot of new things are being used! We should define them in one line!}
% By looking at the above lemmas a natural way to have a small $\wedge-$dt for $f$ would be either to have a better bound on $\mathsf{mbs}(f)$ in terms of $\log(\spar(f))$ or to have a better bound on $\mathsf{fmbs}(f)$ in terms of $\log(\spar(f))$. We address both these approaches.

The above proof technique gives rise to a natural question, can these relationships between monotone and related measures be improved?
The main objective of this article is to explore this question. Specifically, we look at the following variants.
\begin{itemize}
    \item Is it possible to improve the exponent in the relationship $\mathsf{mbs}(f) = O(\log^{2}(\spar(f)))$?
    \item Can we translate the bound on $\mathsf{fmbs}(f)$ by \emph{well behaved quantities} using their bound on $\mathsf{mbs}(f)$. (Similar to the result of Tal~\cite{Tal13} which 
    %showed that the ratio $\fbs(f^{l})/\bs(f^{l})$ is upper bounded by a function of $n$ that doesn't depend on $l$ (the number of iterations), this result 
    allows us to lift upper bounds on block sensitivity to fractional block sensitivity for many well behaved measures.) 
    \item Are there specific class of functions for which monotone analogues have a better dependence on sparsity?
\end{itemize}

Ideally, we would like to compile a table similar to \cite{ABK+} for monotone measures. We start by giving a preliminary table in Appendix~\ref{appendix:table}.

\subsection{Our Results} \label{section 1}
%%%%%%%%%%%%%%%%%%%%%%%%%%%%%%%%%%%%%%%%%%%%%%%%%%%%%%%%%%%%%%%%%%

%Relationship between $\s(f)$ and $\deg(f)$ 
%The relation $\mathsf{mbs}(f)=O(\log^{2}(\spar(f)))$ is proved in \cite{KLMY21} by using the relation $\s(f) = O(\deg(f)^2)$ proved by Nisan and Szegedy \cite{NS94}. In fact, the proof  shows that if for a constant $\alpha$ $\forall f: \s(f) = O(\deg(f)^\alpha)$, then $\forall f: \mathsf{mbs}(f)=O(\log^{\alpha}(\spar(f)))$.
% Interestingly, it happens that improving this aforementioned relationship between $\mathsf{mbs}(f)$ and $\log(\spar(f))$ is in turn improving the relationship between $\s(f)$ and $\deg(f)$.

We study the monotone analogues of standard complexity measures like block sensitivity, certificate complexity and their relations with standard complexity measures. 

%\farzan{Begin change: this will need rewriting to handle the change in order}

%It is natural to ask %as to what happens when $f$ is an arbitrary boolean function? In other words, 
It is natural to ask if it is possible to improve upper bounds on these monotone measures? One very interesting approach for improving bounds on $\fbs(f)$ (standard measure) is by Tal~\cite{Tal13}. He showed that the ratio of $\bs(f^l)$ and $\fbs(f^l)$ is bounded by a quantity independent of $l$; this allowed him to lift \emph{any} upper bound on $\bs(f)$ by a measure which is \emph{well behaved with respect to composition} to $\fbs(f)$. We prove a similar result for $\mathsf{mbs}(f)$ and $\mathsf{fmbs}(f)$. 

%There does not appear to be a direct way answer this and so we check if tight bounds on $\mathsf{mbs}(f)$ w.r.t a ``nice'' complexity measure can be lifted to tight bounds on $\mathsf{fmbs}(f)$.

%To achieve this lift of bounds from $\mathsf{mbs}(f)$ to $\mathsf{fmbs}(f)$ we use the following theorem: 

\begin{theorem}
\label{thm: ratio}
Consider a Boolean function $f:\{0,1\}^{n}\rightarrow\{0,1\}$. For a sufficiently large $n$, the ratio $$\frac{\mathsf{fmbs}(f^{l})}{\mathsf{mbs}(f^{l})}\leq p(n)$$ for all $l\geq 1$, where $p(n)$ is a function in $n$ independent of $l$. 
\end{theorem}

%\rajat{Moved from our technique.}
As mentioned earlier, there is a nice implication of this behaviour (as shown in~\cite{KT16} for standard measures): given a measure $M$ which \emph{behaves well under composition} and an upper bound on monotone block sensitivity in terms of measure $M$, we can lift the same upper bound to fractional monotone block sensitivity.

\begin{corollary}\label{lift_intro}
     Let $f:\{0,1\}^{n}\rightarrow\{0,1\}$ be a Boolean function. Let $M(.)$ be a complexity measure such that for all $l\geq 2,$ $M(f^l)\leq M(f)M(f^{l-1})$. If $\mathsf{mbs}(f)\leq M(f)^{\alpha}$ then $\mathsf{fmbs}(f)=O(M(f)^{2\alpha})$. Furthermore, if $M(1-f)=O(M(f))$ then $\mathsf{fmbs}(f)=O(M(f)^{\alpha})$.
    %Let $M(.)$ be a complexity measure such that $M(1-f) = O(M(f))$ and $\forall~l\geq 2: M(f^l)\leq M(f)M(f^{l-1})$. If for all $f$, $\mathsf{mbs}(f)\leq M(f)^{\alpha}$, then for all $f$, $\mathsf{fmbs}^{0}(f)\leq M(f)^{\alpha}.$
\end{corollary}

We need another extra condition on $M$, $M(1-f) = O(M(f))$, as compared to Kulkarni and Tal~\cite{KT16}. However, most of the complexity measures should satisfy this condition trivially.

% Ideally, we would like to use logarithm of sparsity as the measure $M$ in the above corollary, implying a $\log^2(\spar(f))$ bound on $\mathsf{fmbs}(f)$. Unfortunately, this cannot be done since it is not hard to see that the logarithm of sparsity does not compose in general. Indeed the sparsity of a composed function $f \circ g$ can depend on the degree of $f$ which can be much larger than the logarithm of the sparsity. 
% As a corollary, any upper bound on $\mathsf{mbs}(f)$ by a measure which is well behaved under composition can be lifted to $\mathsf{fmbs}(f)$ (Corollary~\ref{lift}). For a rigorous definition of well behaved measure under composition, please see Section~\ref{beh of monotone}. 

It is tempting to apply this corollary on $\log(\spar(f))$ and try to improve the upper bound on $\mathsf{fmbs}(f)$ in terms of $\log(\spar(f))$. We show a negative result here: $\log(\spar(f))$ does not behave well under composition, indeed the sparsity of a composed function $f \circ g$ can depend on the degree of $f$ which can be much larger than the logarithm of the sparsity.  Hence, Corollary~\ref{lift_intro} can not be used to improve the upper bound on $\mathsf{fmbs}(f)$. For the counterexample, please see Section~\ref{beh of monotone}.

Although our attempt to improve the bound on $\mathsf{fmbs}{(f)}$ did not bear success; we asked, is it possible to improve the log-sparsity upper bound on $\mathsf{mbs}(f)$?
For the question of improving the relation $\mathsf{mbs}(f) = O(\log(\spar(f))^2)$, we show that it will improve the upper bound on sensitivity in terms of degree (a central question in this field).

\begin{theorem}
\label{s(f)-deg}
   If there exists an $\alpha$ s.t. for every Boolean function $f:\zone^{n}\rightarrow\zone$, $\mathsf{mbs}(f)=O(\log^{\alpha}{\spar(f)})$, then for every Boolean function $f:\zone^{n}\rightarrow\zone$ $\s(f)=O(\deg^{\alpha}(f))$.
\end{theorem}

 The converse of this result follows from the proof of $\mathsf{mbs}(f) = O(\log(\spar(f))^2)$ in~\cite{KLMY21}. Nisan and Szegedy~\cite{NS94} showed that $\s(f) = O(\deg(f)^2)$. However, the best possible separation known is due to Kushilevitz (described in \cite{NW94}) giving a function $f$ such that $\s(f) = \Omega(\deg(f)^{1.63})$. So, our result implies that the best possible bound on monotone block sensitivity in terms of logarithm of sparsity cannot be better than $\mathsf{mbs}(f) = O(\log(\spar(f))^{1.63})$.

% We have shown that, the question of the best upper bound by degree on sensitivity is equivalent to the question of best possible upper bound on monotone block sensitivity by $\log(\spar)$ of the function. Note that the relation between $\mathsf{mbs}$ and $\log(\spar)$ was investigated in \cite{KLMY21} and they proved that for all Boolean function $f$, $\mathsf{mbs}(f) = O(\log(\spar(f))^2)$. 

%\subsection{add the theorem statement for the equivalence question, item 1 from abstract} 
%\textcolor{red}{we may want to change the order}\\

Going further, we ask if these bounds can be improved for a class of functions instead of a generic Boolean function? 
Buhrman and de Wolf \cite{buhrman01} proved that the log-rank conjecture holds when the outer function is monotone or symmetric. It turns out that all these monotone measures are equal for these classes of functions. 
% While investigating the monotone measures we have identified some important classes of Boolean function for which almost all the monotone measures are equivalent. 
%\subsection{lemma 4.1 + 4.2}
\begin{theorem}
\label{thm:symm_monotone}
    If $f:\zone^{n}\rightarrow\zone$ is either symmetric or monotone, then
    $$ \ms(f) = \mathsf{mbs}(f) = \mathsf{fmbs}(f) = \hsc(f).$$
\end{theorem}
 This implies an upper bound of $O(\log^2(\spar(f)))$ on $\hsc$ for these functions. For symmetric functions, this bound can be improved to $O(\log \spar(f))$ by combining the above relation with the upper bound on communication complexity for the corresponding AND-functions \cite{buhrman01}. Moreover, Buhrman and de Wolf \cite{buhrman01} showed that $\mathsf{mbs}(f) = \Omega(\log(\spar(f))/\log n)$, which implies that the upper bound is essentially tight.

%This result implies a much better bound of $\log^2(\spar(f))$ on $\mathsf{fmbs}$ and $\log^3(\spar(f)) \log n$ on $D^{cc}(f\circ \wedge_{2})$ when $f$ is symmetric or monotone. The bound on general functions are $\log^4(\spar(f))$ and $\log^5(\spar(f)) \log n$ respectively.

%\farzan{End change 1}

%This result is motivated from a similar result for block sensitivity and fractional block sensitivity by Tal~\cite{Tal13}. The above result shows that as we compose a function multiple times, the ratio of fractional monotone block sensitivity and monotone block sensitivity can't keep increasing (as we would expect if both the measures composed).
\paragraph*{\textbf{Organization:}}
In Section~\ref{preliminaries} we recall the definitions of standard Boolean complexity measures as well as state their monotone analogues.
%that were introduced in \cite{KLMY21}. 
In Section~\ref{section:proof outline} we will give the proof ideas of our results. 
%Precisely, in Section~\ref{beh of monotone} we study the relationships between the monotone complexity measures as well as we show that for every Boolean function $f:\zone^{n}\rightarrow\zone$ the ratio $\frac{\mathsf{fmbs}(f^{l})}{\mathsf{mbs}(f^{l})}$ is independent of $l$. 
This section also contains the counterexample which shows that the relationship between $\mathsf{fmbs}{f}$ and $\log(\spar(f))$ can't be improved using this method (Section~\ref{beh of monotone}). Section~\ref{conclusion} contains the conclusion and some related open problems to pursue. 

Appendix~\ref{appendix:proof of theorem1} and Appendix~\ref{appendix:results for theorem1} contain the complete proof of Theorem~\ref{thm: ratio}. Appendix~\ref{mbs vs logspar} contains the equivalence between the problem of upper bounding $\mathsf{mbs}(f)$ in terms of $\log(\spar(f))$ and the well-studied problem of upper bounding $\s(f)$ in terms of $\deg(f)$ (Theorem~\ref{s(f)-deg}). Finally, Appendix~\ref{mbs-hsc} shows that for the common classes of symmetric and monotone Boolean functions $\hsc$ and $\mathsf{mbs}$ are the same (Theorem~\ref{thm:symm_monotone}). In Appendix~\ref{appendix:table}, we give an overview of the present scenario of the relationships between monotone measures.

%%%%%%%%%%%%%%%%%%%%%%%%%%%%%%%%%%%%%%%%%%%%%%%%%%%%%%%%%%%%%%%%%%
\section{Preliminaries} \label{preliminaries}
%%%%%%%%%%%%%%%%%%%%%%%%%%%%%%%%%%%%%%%%%%%%%%%%%%%%%%%%%%%%%%%%%%

For the rest of the paper, $f$ denotes a Boolean function $f:\zone^{n}\rightarrow\zone$ if not stated otherwise. We start by introducing the following notations that will be used in the paper:
\begin{itemize}
    \item $[n]$ denotes the $\{1,2,...,n\}$. For a set $C\subseteq[n]$, $|C|$ denotes its cardinality.
    \item For a string $x\in\zone^{n}$, its support is defined as $\supp(x):=\{i:x_{i}=1\}$ and $|x|:=|\supp(x)|$ denotes its Hamming weight .
    \item For a string $x\in\zone^{n}$, $x^{\oplus i}$ denotes the string obtained by flipping the $i^{th}$ bit of the string $x$.
    \item For a string $x\in\zone^{n}$ and a $B \subseteq[n]$, $x^{B}$ denotes the string obtained by flipping the input bits of $x$ that correspond to $B$.
    \item Every Boolean function $f:\zone^{n}\rightarrow\zone$ can be expressed as a polynomial over $\mathbb{R}$,
    $f(x)=\sum_{S\subseteq[n]}\alpha_{S}\prod_{i\in S}x_{i}.$
    The sparsity of $f$ is defined as $\spar(f):=|\{S\neq \emptyset:\alpha_{S}\neq 0\}|$ and the degree of $f$ is defined as $\deg(f):=\underset{S\subseteq [n]: \alpha_{S}\neq 0}{\max}|S|$. %For a set $S \subseteq [n]$, the monomial $\prod_{i \in S} x_i$ will be denoted by $X_S(x)$. 
    % \item For an input $x\in\zone^{n}$ and a Boolean function $f:\zone^{n}\rightarrow\zone$, $W(f,x):=\{B\subseteq[n]:f(x^{B})\neq f(x)\}$.  
\end{itemize}

% Let $f:\{0,1\}^{n}\rightarrow\{0,1\}$ be a Boolean function. For an input $x\in\{0,1\}^{n}$, we say a subset $B\subseteq[n]$ is sensitive for $x$ if $f(x^{B})\neq f(x)$, where $x^{B}$ is obtained by flipping all the bits of $x$ that corresponds to $B$.

% In this section, we recall the definitions of some of the standard Boolean complexity measures along with their monotone analogues which were introduced in \cite{KLMY21}.

Having introduced the notations, we now recall the definitions of standard Boolean complexity measures. 
\begin{definition}[Sensitivity]
\label{defi: s}
For an input $x\in\zone^{n}$ the $i^{th}$ bit is said to be sensitive for $x$ if $f(x^{\oplus i}) \neq f(x)$. The sensitivity of $x$ w.r.t $f$ is defined as $$\s(f,x) := |\{i \in [n]: f(x^{\oplus i}) \neq f(x)\}|,$$ while the sensitivity of $f$ is defined as $$\s(f) := \underset{x\in\zone^{n}}{\max}\s(f,x).$$
\end{definition}

\begin{definition}[Block Sensitivity]
\label{defi: bs}
For an input $x\in\zone^n$, a subset $B\subseteq[n]$ is said to be a sensitive block for $x$ w.r.t $f$ if $f(x^B)\neq f(x).$ 
The block sensitivity of $f$ at $x$, denoted by $\bs(f,x)$, is defined as 
$$\bs(f,x)=\max\{k | \exists B_{1}, \dots ,B_{k} \textit{ with } B_{i}\cap B_{j}=\emptyset \textit{ for } i\neq j \textit{ and } f(x^{B_{i}})\neq f(x)  \}.$$

Block sensitivity of $f$ is defined as:
$$\bs(f):=\underset{x\in\{0,1\}^{n}}{\max} \bs(f,x).$$
\end{definition}

Fractional block sensitivity ($\fbs$) is obtained by allowing fractional weights on sensitive blocks.

% On allowing fractional weights for sensitive blocks, we get the fractional block sensitivity ($\fbs$).  
% The block sensitivity of a function $f$ at an input $x$ can be alternatively defined as an integer linear program. We will be using the linear programming formulations for our proof. Let $W(f,x):=\{B\subseteq [n]:f(x^{B})\neq f(x)\}$ denote the set of all sensitive blocks for the input $x\in\{0,1\}^{n}$ then:
% $$\bs(f,x):=\max \sum_{w\in W(f,x)}b_{w}$$
% s.t. $$\forall i\in[n],\sum_{w\in W(f,x):i\in w}b_{w}\leq 1$$
% and $$\forall w\in W(f,x),\; b_{w}\in\{0,1\}$$

% If we relax the integrality condition, $b_{w}\in\{0,1\}$ to $b_{w}\in[0,1]$ the corresponding value of the objective function is referred to as the fractional block sensitivity of $f$ at $x$ and is denoted by $\fbs(f,x)$. 
\begin{definition}[Fractional Block Sensitivity]
\label{defi: fbs}
Let $W(f,x):=\{B\subseteq [n]:f(x^{B})\neq f(x)\}$ denote the set of all sensitive blocks for the input $x\in\{0,1\}^{n}$.
The fractional block sensitivity of f at x, denoted by $\fbs(f,x)$ is the value of the linear program:\\
$$\fbs(f,x):= \max \sum_{w\in W(f,x)}b_{w}$$\\
s.t. $$\forall i\in[n],\sum_{w\in W(f,x):i\in w}b_{w}\leq 1$$\\
and $$\forall w\in W(f,x),\; b_{w}\in [0,1].$$\\
The fractional block sensitivity of $f$ is defined as:
$$\fbs(f):=\underset{x\in\zone^{n}}{\max}\fbs(f,x).$$

\end{definition}
Note that restricting the linear program for $\fbs(f,x)$ to only integral values gives $\bs(f,x)$.
%\textcolor{blue}{
% The certificate complexity of $f$ at $x$ is defined to be the size of the smallest certificate for $x$. Formally,}
\begin{definition}[Certificate Complexity]\label{defi:C}
For a function $f$ and an input $x\in\{0,1\}^{n}$, a subset $C\subset[n]$ is said to be a certificate for $x$ if for all $y\in\{y\in\{0,1\}^{n}:\forall i\in C,\; x_{i}=y_{i}\}$ we have $f(x)=f(y)$.
For a function $f$ and an input $x \in \{0,1\}^{n}$ the certificate complexity of $f$ at $x$, denoted by $\cert(f,x)$, is defined as:
    $$\cert(f,x):=\underset{C:\text{$C$ is a certificate for $x$}}{\min}|C|.$$

    The certificate complexity of $f$ is defined as:
    $$\cert(f):=\underset{x\in\{0,1\}^{n}}{\max} \cert(f,x).$$
\end{definition}

% \textcolor{red}{blocks are $B_1$ or $\mathcal{B}_1$? stick to anyone of them}
% Like block sensitivity of $f$ at an input $x$ i.e. $\bs(f,x)$ we can also describe the certificate complexity of $f$ at an input $x$ as a linear program.
% Formally, if $\mathcal{B}_{1},..,\mathcal{B}_{k}$ is the collection of all minimal sensitive blocks for $f$ at $x$ then $\cert(f,x)$ is the solution for the following minimization integer linear program:
% $$\min \sum_{i\in[n] }c_{i}$$
% s.t.
% $$\forall j\in[k],\; \sum_{i:i\in \mathcal{B}_{j}}c_{i}\geq 1,$$
% $$\forall i\in[n],\; c_{i}\in\{0,1\}.$$
% On relaxing the integrality condition for $c_{i}$s i.e. $\forall i\in[n],\; c_{i}\in[0,1]$ the solution of the linear program is the Fractional Certificate Complexity of $f$ at $x$ and is denoted as $\fcert(f,x)$.

The fractional measures $\fbs$ and $\fcert$ were introduced in \cite{Tal13}. There it was observed that for all $x\in\zone^{n}$ we have :
$\fbs(f,x)=\fcert(f,x)$
since the linear program for Fractional Certificate Complexity and Fractional Block Sensitivity are the primal-dual of each other and are also feasible.

For each of these standard measures, the analogous monotone versions can be defined by restricting functions $f$ to the positions in the support of a given input $x\in\{0,1\}^{n}$. Formally, for a function $f:\{0,1\}^{n}\rightarrow\{0,1\}$ and an input $x\in\{0,1\}^{n}$ let $f_{x}$ denote the function $f$ obtained by restricting $f$ to the set $\{y\in\{0,1\}^{n}:\forall i\in \supp(x),\; y_{i}=1\}$.

\begin{definition}[Monotone Sensitivity]
\label{defi:ms}
The monotone sensitivity for $x$ is defined as $\ms(f,x):=\s(f_{x},0^{n-|x|})$ while the monotone sensitivity for $f$ is defined as $$\ms(f):=\underset{x\in\zone^{n}}{\max}\ms(f,x).$$
\end{definition}
\begin{definition}[Monotone Block Sensitivity]
\label{defi: mbs}
The monotone block sensitivity of a function $f$ at an input $x\in\{0,1\}^{n}$ is defined as $\mathsf{mbs}(f,x):=\bs(f_{x},0^{n-|x|})$ while the monotone block sensitivity of $f$ is defined as:
$$\mathsf{mbs}(f)=\underset{x\in\{0,1\}^{n}}{\max}\mathsf{mbs}(f,x).$$
\end{definition}

Similar to block sensitivity, fractional block sensitivity of $f$ can be extended to the monotone setting by defining the linear program over the sensitive monotone blocks i.e. sensitive blocks containing only $0$'s.

\begin{definition}[Fractional Monotone Block Sensitivity]
\label{defi: fmbs}
For a function $f$ the fractional monotone block sensitivity at an input $x\in\{0,1,\}^{n}$ is defined as: $\mathsf{fmbs}(f,x):=\fbs(f_{x},0^{n-|x|})$ and the fractional monotone block sensitivity of $f$ is defined as:
$$\mathsf{fmbs}(f):=\underset{x\in\{0,1\}^{n}}{\max}\mathsf{fmbs}(f,x).$$
\end{definition}

Certificate complexity can also be extended to the monotone setting by counting only the zero entries in the certificate. The monotone analogue of certificate complexity was introduced in \cite{KLMY21} as hitting set complexity (it can be viewed as a hitting set for system of monomials). Formally,
\begin{definition}[Monotone Certificate Complexity/Hitting Set Complexity]
\label{defi: HSC}
For a function $f$ and an input $x\in\{0,1\}^{n}$ the hitting set complexity for $x$ is defined as:
$$\hsc(f,x):=\cert(f_{x},0^{n-|x|}),$$
while the hitting set complexity of the function $f$ is defined as:
$$\hsc(f):=\underset{x\in\{0,1\}^{n}}{\max} \hsc(f,x).$$
\end{definition}

Since $\bs$ allows only integer solutions to $\fbs$ linear program, and $\cert$ only allows integer solutions to the dual linear program~\cite{Tal13},
$$\bs(f_{x},0^{n-|x|})\leq \fbs(f_{x},0^{n-|x|}) \leq \cert(f_{x},0^{n-|x|}),$$
By similar arguments,
$$\mathsf{mbs}(f,x)\leq \mathsf{fmbs}(f,x) \leq \hsc(f,x).$$

% Till now we have defined the above measures with respect to a particular input but they can also be defined for the function $f$ itself in the following manner:

% For a complexity measure $M \in \{\s, \bs, \fbs, \cert, \ms, \mathsf{mbs}, \mathsf{fmbs}, \hsc\}$ and $b \in \{0,1\}$,  $$M(f) := \underset{x\in\zone^{n}}{\max} M(f, x).$$
Instead of taking maximum over all inputs, these measures can be defined for a certain output too. In other words, for a complexity measure $M \in \{\s, \bs, \fbs, \cert, \ms, \mathsf{mbs}, \mathsf{fmbs}, \hsc\}$ and $b \in \{0,1\}$,  $$M^b(f) := \underset{x\in f^{-1}(b)}{\max} M(f, x).$$

%%%%%%%%%%%%%%%%%%%%%%%%%%%%%%%%%%%%%%%%%%%%%%%%%%%%%%%%%%%%%%%%%%
\section{Proof Outline}
\label{section:proof outline}
%\subsection{Proof outline}

 First, we outline the ideas for the proofs of \autoref{s(f)-deg} and \autoref{thm:symm_monotone}. Subsequently, we will give proof outline for our main result, \autoref{thm: ratio}.
 %which encompasses additional corollaries and lemmas.

%For readability first we will present the proof idea of \autoref{s(f)-deg} and \autoref{thm:symm_monotone}. Then we will move to our main result \autoref{thm: ratio} which comprises some other corollaries and lemma. 

\paragraph*{Proof idea of Theorem \ref{s(f)-deg}}

We would like to prove that $\s(f)=O(\deg^{\alpha}(f))$ for any Boolean function $f$ (given that $\mathsf{mbs}(g)=O(\log^{\alpha}{\spar(g)})$ for all Boolean functions $g$). The idea is to convert $f$ into $\tilde{f}$ by shifting the point with maximum sensitivity to $0^n$; this transformation can only decrease the degree and $\mathsf{mbs}(f)$ is higher than $\s(\tilde{f})$.

The rest is accomplished by using the fact that sparsity is at most exponential in degree. This is shown for Boolean functions with $\{-1,1\}$ domain first using Parseval's identity, and then it can be translated for Boolean functions with $\{0,1\}$ domain.

For the interest of space we will present the proof of \autoref{s(f)-deg} in Appendix~\ref{mbs vs logspar}.

\paragraph*{Proof idea of Theorem \ref{thm:symm_monotone}}
We deal with the cases of monotone and symmetric boolean functions separately.

For monotone boolean functions, the idea for showing equality between the monotone versions of the standard boolean complexity measures is similar to the approach used for the standard complexity measures i.e. we consider a string $x$ which achieves the hitting set complexity $\hsc(f)=\hsc(f,x)$ with $C$ as one of its witness. 
Now, using $x$ and $C$ we construct another input $x^{'}$ with $\supp(x)\subseteq \supp(x^{'})$ and $\supp(x^{'})\cap C=\phi$ s.t. every bit $i\in C$ is sensitive for $f_{x}$ at $x^{'}$. Hence leading to $\hsc(f,x)\leq \ms(f,x^{'})\leq \ms(f)$.

Now for the case of symmetric boolean functions, we show that there exists an input $z\in\{0,1\}^{n}$ s.t. $\hsc(f)=\hsc(f,z)$ and $\hsc(f,z)=n-|z|$, where $|z|$ is the Hamming weight of $z$ i.e. $|z|=\supp(f)$. But this implies $\ms(f,z)=\s(f_{z},0^{n-|z|})=n-|z|=\hsc(f,z)=\hsc(f)$. 

We will give a complete proof in Appendix~\ref{mbs-hsc}.

\subsection{$\mathsf{fmbs}$ versus $\mathsf{mbs}$}\label{beh of monotone}

 % For symmetric and monotone boolean functions the result follows due to Theoremwe have a stronger relation i.e. the ratio $\frac{\mathsf{fmbs}(f^{l})}{\mathsf{mbs}(f^{l})}=1$, which follows from Theorem \ref{thm:symm_monotone}.

Let us move to the proof idea of \autoref{thm: ratio}, we essentially follow the same proof outline as~\cite{Tal13}. 
\autoref{thm: ratio} proves that for any $f:\zone^n \rightarrow \zone$, the ratio $\frac{\mathsf{fmbs}(f^l)}{\mathsf{mbs}(f^l)}$ is bounded above by a function of just $n$ (and independent of $l$). In other words, composition make $\mathsf{fmbs}$ and $\mathsf{mbs}$ equal in the asymptotic sense.  %Theorem \ref{thm: ratio} answers this question in affirmative. 

We will be considering the case of monotone functions and non-monotone functions separately. While the case for monotone functions is handled easily due to \autoref{thm:symm_monotone} (we have a stronger relation $\frac{\mathsf{fmbs}(f^{l})}{\mathsf{mbs}(f^{l})}=1$), most of the work is done for the case when $f$ is non-monotone.

~~\\

\paragraph*{Proof outline of Theorem \ref{thm: ratio}} ~~\\

From the discussion above, assume that $f:\zone^n \rightarrow \zone$ is a non-monotone Boolean function. We want to show that 
$$ \mathsf{mbs}(f^{l})\geq p(n) ~\mathsf{fmbs}(f^{l}), $$
for some function $p(n)$ and big enough $l$. 

Similar to $\fbs(f)$, $\mathsf{fmbs}(f)$ can also be written as a fractional relaxation (linear program) of an integer program for $\mathsf{mbs}(f)$. The proof converts a feasible solution of the linear program for $\mathsf{fmbs}(f^{l+1})$ into a feasible solution of $\mathsf{mbs}(f^{l+1})$ without much loss in the objective value, bounding $\frac{\mathsf{mbs}(f^{l+1})}{\mathsf{fmbs}(f^{l+1})}$ in terms of $\frac{\mathsf{mbs}(f^{l})}{\mathsf{fmbs}(f^{l})}$: 
 %with the guarantee that $\mathsf{fmbs}(f^{l+1})$ is not much larger than $\mathsf{mbs}(f^{l+1})$, 
 \begin{equation} \label{eq:bounded_ratio}
\mathsf{mbs}(f^{l+1})\geq \mathsf{fmbs}(f^{l+1})\frac{\mathsf{mbs}(f^{l})}{\mathsf{fmbs}(f^{l})}\alpha_l,
\end{equation}
where $\alpha_{l}$ s.t. $\prod_{l=1}^{\infty}\alpha_{l}=\Omega(1)$. This finishes the proof by taking large enough $l$. We are left with proving \autoref{eq:bounded_ratio} for some $\alpha_l$'s. 

%Hence for any $l_{0}\in\mathbb{Z}$ we have:
%$$\frac{\mathsf{mbs}(f^{l_{0}})}{\mathsf{fmbs}(f^{l_{0}})}\geq \bigg(\prod_{l=1}^{l_{0}-1}\alpha_{l}\bigg)\frac{\mathsf{mbs}(f)}{\mathsf{fmbs}(f)}.$$

Remember, the idea is to convert a solution of $\mathsf{fmbs}(f^{l+1})$ into a solution of $\mathsf{mbs}(f^{l+1})$. Let $x:=(x^{1},x^{2},...,x^{n})\in\{0,1\}^{n^{l+1}}$ be the input s.t. $\mathsf{fmbs}(f^{l+1})=\mathsf{fmbs}(f^{l+1},x)$ where $x^{1},x^{2},...,x^{n}\in\{0,1\}^{n^{l}}$. The input $y\in\zone^{n}$ be the $n$-bit string corresponding to $x$ i.e.$\; \forall\; i\in[n],\; y_{i}:=f^{l}(x^{i})$. We know that $f^{l+1}(x) = f(y)$.

Let $\{B_{1},...,B_{k}\}$ be the set of all minimal monotone blocks for $f$ at $y$. A minimal monotone block, say $B = \{i_1,i_2,\cdots,i_k\}$, of $y$ gives minimal monotone blocks for $f^l$ at inputs $x^{i_1},x^{i_2}, \cdots, x^{i_k}$. Observe that the total weight contributed by any block $B_i$ in the linear program for $f^{l+1}$ will become feasible for the following linear program:
$$\max \sum_{i=1}^{k}w_{i},$$
s.t.
$$\sum_{j:i\in B_{j}}w_{j}\leq \mathsf{fmbs}(f^{l},x^{i}),\; \forall\; i\in[n],$$
$$w_{j}\geq 0, \forall j\in[k].$$

A small modification to these weights (multiplying by a quantity closely related to $\frac{\mathsf{mbs}(f^{l})}{\mathsf{fmbs}(f^{l})}$ and taking their integer part) gives the solution of the following integer program (notice that $\mathsf{mbs}(f^l)$ is taken over another suitable input $\hat{x}$):
$$\max \sum_{i=1}^{k}w_{i},$$
s.t.
$$\sum_{j:i\in B_{j}}w_{j}\leq \mathsf{mbs}(f^{l},\hat{x}^{i}),\; \forall\; i\in[n],$$
$$w_{j}\in \{0,1,2, \dots,\mathsf{mbs}(f^{l})\}, \forall j\in[k].$$

Let $\{w^{'}_i\}$ be the solution of the program above. Using this assignment $w^{'}_{i}$ we can construct $\sum_{i=1}^{k}w^{'}_{i}$ many disjoint monotone sensitive blocks of $f^{l+1}$ (see Appendix~\ref{appendix:proof of theorem1} for this construction). 

It can be shown that the objective value of the obtained solution satisfies, 
$$\mathsf{mbs}(f^{l+1}) \geq \mathsf{fmbs}(f^{l+1}) \frac{\mathsf{mbs}(f^{l})}{\mathsf{fmbs}(f^{l})}  - 2^{n}.$$
Here, the term $2^n$ appears because we take the integer part of a fractional solution to construct $\{w^{'}_i\}$. This inequality can be converted into \autoref{eq:bounded_ratio} by using properties of composition of fractional monotone block sensitivity and some minor assumptions on $\mathsf{mbs}(f)$. We present the complete proof of \autoref{thm: ratio} in Appendix~\ref{appendix:proof of theorem1}.

%Hence by \autoref{thm: ratio} we have that for every Boolean function $f:\zone^{n}\rightarrow\zone$ $\mathsf{fmbs}(f^{l})/\mathsf{mbs}(f^{l})\leq p(n)$ for all $l\in\mathbb{N}$. 
\paragraph*{Implications of \autoref{thm: ratio}} ~~\\

 One of the reason \autoref{thm: ratio} is interesting because it provides a way of lifting upper bounds on $\mathsf{mbs}(f)$ to upper bounds on $\mathsf{fmbs}(f)$. This was observed by \cite{KT16} for the standard setting ($\bs$ and $\fbs$), using which they showed the quadratic relation between $\fbs(f)$ and $\deg(f)$ i.e. $\fbs(f)=O(\deg^2(f))$. This was an improvement over $\fbs(f)\leq \cert(f)=O(\deg^{3}(f))$~\cite{fbsdeg}.

Similarly, we can do the lifting for $\mathsf{mbs}(f)$ and $\mathsf{fmbs}(f)$ which we have stated in \autoref{lift_intro}. We present the proof of \autoref{lift_intro} in \autoref{appendix:proof of theorem1}.

We now give an example showing that $\log \spar(f^2)$ may be exponentially larger than $(\log \spar(f))^2$ and so \autoref{lift_intro} cannot be applied to $\log \spar$. For any Boolean functions $f$ and $g$, $\spar(f \circ g) \geq (\spar(g)-1)^{\deg(f)}$ (see, for instance, \cite{loff19}\footnote{Their proof is stated for sparsity in the Fourier representation, but is readily seen to work for block composition of arbitrary multilinear polynomials.}). In particular, when $\spar(g) \geq 3$, $\log \spar(f \circ g) \geq \deg(f)$.
So any function $f$ satisfying $\spar(f) \geq 3$ and $\deg(f) = 2^{\Omega(\log \spar(f))}$ gives us the desired separation. For instance, we may take, \\
$$f(x_1, x_2, \dots, x_{n}) = \OR(\AND(x_1, x_2, \dots, x_{n/2}), \AND(x_{n/2+1}, x_{n/2+2}, \dots, x_{n}))$$ which has degree $n$ and sparsity only $3$.

%%%%%%%%%%%%%%%%%%%%%%%%%%%%%%%%%%%%%%%%%%%%%%%%%%%%%%%%%%%%%%%%%%

\section{Conclusion}
\label{conclusion}

In the present work we studied the behaviour of different monotone complexity measures and their relation with one another. The relations between these measures are natural questions by themselves; on top of that, they can potentially be used to improve the upper bound on deterministic communication complexity in terms of logarithm of sparsity.

To summarize our results, we were able to show a better upper bound on $\hsc(f)$ in terms of $\log(\spar(f))$ for monotone and symmetric Boolean functions. It will be interesting to find other class of functions for which the upper bound can be improved. Our result that the $\mathsf{mbs}$ vs. $\log(\spar)$ question is equivalent to the $\s$ vs. $\deg$ question, might give another direction to attack this old open question. 

This work also showed that the ratio $\frac{\mathsf{fmbs}(f^{l})}{\mathsf{mbs}(f^{l})}$ is independent of the iteration number $l$. Even though we were not able to use it to show $\mathsf{fmbs} = O(\log (\spar(f))^2)$, this results seems to be of independent interest in terms of behavior of these monotone measures. %It remains open to prove/disprove whether $\log(\spar(f))$ is a ``nice'' measure according to \autoref{lift}.

% Here, we have shown that, the question of the best possible upper bound on monotone block sensitivity by logarithm of sparsity is equivalent to the natural question of best upper bound by degree on sensitivity (standard measures), which is a well studied problem in this area. We also tried to address whether $\mathsf{fmbs}(f)=O(\log^{2}(\spar(f)))$? The best known bound on $\mathsf{fmbs}$ in terms of $\log(\spar)$ is $\mathsf{fmbs}(f) = O(\log^{4}(\spar(f)))$ for all Boolean function $f$. Our approach for answering the aforementioned problem was based on extending the technique introduced in \cite{KT16} to $\mathsf{fmbs}(f)$ and $\mathsf{mbs}(f)$.

Some of the other open questions from this work are listed below.

\begin{open question}
Can we prove that for any Boolean function $f:\zone^n \rightarrow \zone$, $\mathsf{fmbs}(f) = O(\log (\spar(f))^2)$?
\end{open question}

Another possible open question in this direction is asking about best possible separation between $\mathsf{fmbs}$ and $\log(\spar)$. Right now it is known that for all Boolean function $\mathsf{fmbs}(f) = O(\log(\spar(f))^4)$ and the best known separation is due to Kushilevitz (described in \cite{NW94}), giving a function $f$ such that $\s(f) = \Omega(\deg(f)^{1.63})$. Can we give a better separation for monotone measures? 

\begin{open question}
  Does there exist a function $f$ for which, $\mathsf{fmbs}(f) = \Omega(\log(\spar(f))^{\alpha})$ for some $\alpha > 1.63$? 
\end{open question}
                                                           
% Using the lifting theorem for $\mathsf{fmbs}^{0}(f)$, we can obtain a lifting theorem for $\mathsf{fmbs}(f)$:
% \begin{corollary}\label{lifting fmbs}
%     Let $f:\zone^{n}\rightarrow\zone$ be a Boolean function. Let $M(.)$ be a complexity measure s.t. $\mathsf{mbs}(f)=O(M(f)^{\alpha})$. If $M(f^{l})\leq M(f)^{l}$ for all $l\geq 1$ then $\mathsf{fmbs}(f)=O(M(f)^{2\alpha})$. Furthermore, if $M(1-f)=O(M(f))$ then $\mathsf{fmbs}(f)=O(M(f)^{\alpha})$
% \end{corollary}
% \begin{proof}
 
% \end{proof}

%  .
% \end{itemize}

%%%%%%%%%%%%%%%%%%%%%%%%%%%%%%%%%%%%%%%%%%%%%%%%%%%%%%%%%%%%%%%%%%
%\section{Observations that we have}
%%%%%%%%%%%%%%%%%%%%%%%%%%%%%%%%%%%%%%%%%%%%%%%%%%%%%%%%%%%%%%%%%%

% \begin{itemize}

% \item For any Boolean function there exists a $l_{0}\in\mathbb{N}$ s.t. for all $l\geq l_{0}$ $\mathsf{fmbs}(f^{l})\leq \mathsf{mbs}^{1+c}(f^{l})$ for any constant $c> 0$.
% \end{itemize}

% Now, if $s(g,0^{n})<n/3$ for all $g\in\mathcal{G}_{0}$ then it contradicts the above inequality. Hence proving the existence of $g\in\mathcal{G}_{0}$ with $s(g)\geq n/3$. If we assume that $d:=rdeg(OR_{n})\leq n^{\alpha}$ for some $1/2\leq \alpha\leq 1$ then 
% $$s(g)\geq \Omega(d^{1/\alpha})\geq \Omega(deg(g)^{1/\alpha}),$$
% for some $\alpha\in[1/2,1].$

% %\newpage
% %\bibliographystyle{alpha}

\bibliography{arxiv}
\appendix
\section{Proof of ~\autoref{thm: ratio} and \autoref{lift_intro} }
\label{appendix:proof of theorem1}

\begin{proof}[Proof of ~\autoref{thm: ratio}]

     If $f:\zone^{n}\rightarrow \zone$  is monotone then $f^l$ is also monotone. Hence from ~\autoref{thm:symm_monotone} it follows that $\mathsf{mbs}(f^l) = \hsc(f^l)$ which gives $\mathsf{mbs}(f^l) = \mathsf{fmbs}(f^l)$.

    Now we consider $f$ to be a non-monotone Boolean function. We consider two sub cases  $\mathsf{mbs}(f)=1$ and $\mathsf{mbs}(f)\geq 2$ separately. The sub case of $f$ being non-monotone with $\mathsf{mbs}(f)=1$ does not arise in the proof of $\fbs$ and $\bs$ ratio. This is because $\bs(f)\geq 2$ for every non-monotone function $f$. 
% (We provide the characterization of non-monotone Boolean functions with $\mathsf{mbs}=1$ in Appendix~\ref{char of mbs=1}.) 
% For the remaining case of $f$ being non-monotone and $\mathsf{mbs}(f)\geq 2$ we handle it similarly as it was done in \cite{Tal13}.
    
    If $\mathsf{mbs}(f)=1$ and if $\mathsf{mbs}(f^{l})=1$ for all $l\geq 1$ then using the fact that $\mathsf{fmbs}(f)=O(\mathsf{mbs}^{2}(f))$ we get $\mathsf{fmbs}(f^{l})/\mathsf{mbs}(f^{l})=O(1)$. If the aforementioned condition does not hold i.e. there exists a $k\in\mathbb{N}$ s.t. $\mathsf{mbs}(f^{k})\geq 2$ then what remains to show is that $\mathsf{fmbs}(f^{l})/\mathsf{mbs}(f^{l})\leq p(n)$ for all $l\geq k$. It follows that the argument for this part is similar to the case when $f$ is non monotone and $\mathsf{mbs}(f)\geq 2$.

    To prove the theorem for non-monotone functions and $\mathsf{mbs}(f)\geq 2$, we will need several lemmas about the behaviour of these monotone complexity measures under composition.

%In the present section we see the relationship between the monotone complexity measures in more detail. 
\begin{lemma}[\cite{Tal13}]
\label{tal's}
For Boolean functions $f:\zone^{n}\rightarrow\zone$ and $g:\zone^{m}\rightarrow\zone$, if $f(z^{n})=g(z^{m})=z$ for $z\in\zone$ then:
$$\fbs(f\circ g,z^{nm})\geq\fbs(f,z^{n})\fbs(g,z^{m}).$$
\end{lemma}

The above observation can be adapted to $\mathsf{fmbs}^{0}$.
\begin{lemma}\label{fmbs composes}
If $f:\zone^{n}\rightarrow\zone$ and $g:\zone^{m}\rightarrow\zone$ are Boolean functions with $\mathsf{fmbs}^{0}(f)=\mathsf{fmbs}(f,x)$ and $\mathsf{fmbs}^{0}(g)=\mathsf{fmbs}(g,y)$ then:
$$\mathsf{fmbs}^{0}(f\circ g)\geq \mathsf{fmbs}^{0}(f)\mathsf{fmbs}^{0}(g).$$
\end{lemma}
\begin{proof}
Consider the inputs $x\in\zone^{n}$ and $y\in\zone^{m}$ s.t. $\mathsf{fmbs}^{0}(f)=\mathsf{fmbs}(f,x)$ and $\mathsf{fmbs}^{0}(g)=\mathsf{fmbs}(g,y)$. As $f_{x}(0^{n-|x|})=g_{y}(0^{n-|y|})=0$ hence by \autoref{tal's} it follows that:
$$\mathsf{fmbs}(f_{x}\circ g_{y},\boldsymbol{0})\geq \mathsf{fmbs}^{0}(f)\mathsf{fmbs}^{0}(g),$$
where $\boldsymbol{0}$ is the all zero string in $\zone^{(n-|x|)(n-|y|)}$.

Fix any $z \in g^{-1}(1)$. (If $g$ is the constant $0$ function, then the lemma holds since $\mathsf{fmbs}^0(g) = 0$.) Now, consider the input $\gamma:=(\gamma_{1},\gamma_{2},...,\gamma_{m})$ with $\gamma_{1},...,\gamma_{n}\in\zone^{m}$ defined as:
$$\gamma_{i}:=\begin{cases}
    & z,\; if\; x_{i}=1\\
    & y,\; otherwise
\end{cases}$$
Observe that $\mathsf{fmbs}(f\circ g,\gamma)\geq \mathsf{fmbs}(f_{x}\circ g_{y},\boldsymbol{0})$, hence giving us the result:
$$\mathsf{fmbs}^{0}(f\circ g)\geq \mathsf{fmbs}^{0}(f)\mathsf{fmbs}^{0}(g).$$
\end{proof}

The remaining lemmas given below are proved in the Appendix~\ref{appendix:results for theorem1}.

\begin{lemma}
\label{mbs and fmbs of composition}
   Let $f, g$ be two Boolean function where $f$ is non-monotone and $z \in \{0,1\}$ then,\\
           1. $\mathsf{mbs}^z(f \circ g) \geq \mathsf{mbs}(g)$,\\
           2. $\mathsf{fmbs}^z(f \circ g) \geq \mathsf{fmbs}(g)$.
\end{lemma}

\begin{lemma}\label{monotonicity under composition}
    For Boolean functions $f:\zone^{n}\rightarrow\zone$ and $g:\zone^m\rightarrow\zone$ we have:
$$\mathsf{mbs}^z(f\circ g)\geq \max\{\mathsf{mbs}^z(f)\mathsf{mbs}^0(g),\bs^z(f)\min\{\mathsf{mbs}^0(g),\mathsf{mbs}^1(g)\}\}.$$
\end{lemma}

\begin{corollary}
\label{lemma: fmbs monotone increasing}
    Let $f$ be a non-monotone Boolean function with $z \in \{0,1\}$ then, the sequence $\{\mathsf{mbs}^{z}(f^l)\}_{l \in \N}$ is monotone increasing and if $\mathsf{mbs}(f)\geq 2$ then for every $z \in \{0,1\}$ the sequence $\{\mathsf{mbs}^{z}(f^l)\}_{l \in \N}$ tends to infinity.
\end{corollary}

We are now in a position to prove \autoref{thm: ratio}. To recall, \autoref{thm: ratio} states that
for any function $f:\zone^n \to \zone$ the ratio $\frac{\mathsf{fmbs}(f^l)}{\mathsf{mbs}(f^l)}$ is independent of $l$.
    
    Remember that we are left with the case when $f$ is not monotone and $\mathsf{mbs}(f) \geq 2$. If we show that there exists a sequence $\{r_{l}\}_{l\geq 1}$  s.t. for all $l\geq 1$ we have:
    $$\frac{\mathsf{mbs}(f^{l})}{\mathsf{fmbs}(f^{l})}\geq r_{l}\geq 1/p(n),$$
    then we will be done.

    Now consider the sequence: $$r_{l}:=\min\{r^{0}_{l},\; r^{1}_{l}\},$$ where $r^{z}_{l}:=\frac{\mathsf{mbs}^{z}(f^{l})}{\mathsf{fmbs}^{z}(f^{l})}$ for $z\in\zone$.  Taking $z' \in \{0, 1\}$ as $\mathsf{fmbs}(f^l) = \mathsf{fmbs}^{z'}(f^l)$, we get:

    $$ \frac{\mathsf{mbs}(f^l)}{\mathsf{fmbs}(f^l)} = \frac{\mathsf{mbs}(f^l)}{\mathsf{fmbs}^{z'}(f^l)} \geq \frac{\mathsf{mbs}^{z'}(f^l)}{\mathsf{fmbs}^{z'}(f^l)}  = r_l^{z'} \geq r_l,$$
    i.e. $\frac{\mathsf{mbs}(f^{l})}{\mathsf{fmbs}(f^{l})}$ has $r_{l}$ as its lower bound. What remains to show is that for all $l\geq 1$:
    $$r_{l}\geq 1/p(n).$$
 
  Now, notice it is sufficient to show that for $l\geq l_{0}$, $r_{l}\geq 1/p(n),$ where $l_{0}$ is a parameter we fix later. 
%   Now by part 4 of Lemma \ref{bs and fbs composition} and by \ref{lemma: fmbs monotone increasing} it follows that there exist minimal $m,\; m^{'} \in \N$ such that $\bs(f^m) \geq 2.2^n$ and $\mathsf{mbs}(f^{m^{'}})\geq 2.2^{n}$ respectively. Using the given condition that $m$ is minimal such integer and by Theorem \ref{theo: bs of composition}, we have:

%     $$\bs^z(f^{m^{'}+1}) \leq \bs(f^{m^{'}+1})  \leq \fbs(f^{m^{'}-m+2}).\bs(f^{m-1})  \leq n^{m^{'}-m+2}.2.2^n ,           $$
% for all $z\in\zone$.

  % Now, using the fact that $ms^{1-z}(\tilde{f})\leq s^{1-z}(\tilde{f})\leq bs^{1-z}(\tilde{f})=bs^{1-z}(f)$ where $\tilde{f}(x):=f(1-x_{1},1-x_{2},...,1-x_{n})$ and by Lemma \ref{HSC vs mbs and s}, we have that for all $l \in \N$:
  %    $$   \mathsf{fmbs}^{z}(f^{l})\leq HSC^z(f^l) \leq \mathsf{mbs}^z(f^l). s^{1-z}(f^l)  \leq \mathsf{mbs}^z(f^l). \bs^{1-z}(f^l) .  $$  
     
   % As $fmbs(f^{l},x)\leq HSC(f^{l},x)$ for all input $x$, it follows that $$ \mathsf{fmbs}^z(f^l) \leq  HSC^z(f^l)  \leq \mathsf{mbs}^z(f^l). \bs^{1-z}(f^l). $$   

    % Let us consider the case when $l \leq m^{'}+1$. For the aforementioned case we have:
    % $$ r_l^z = \frac{\mathsf{mbs}^z(f^l)}{\mathsf{fmbs}^z(f^l)} \geq  \frac{1}{\bs^{1-z}(f^l)} \geq  \frac{1}{\bs^{1-z}(f^{m^{'}+1})} \geq \frac{1}{n^{m^{'}-m+2}.2.2^n }.$$

To this effect, we show that for $l \geq l_{0}$:
\begin{equation} \label{eq:bound_on_r}
        r_{l+1}\geq r_{l}(1-2^{-1-\lfloor\frac{l-(l_{0}+1)}{2}\rfloor}).
\end{equation}
    
\autoref{eq:bound_on_r} will complete the proof because it implies that for all $s\geq l_{0}$:
    $$r_{s}\geq r_{l_{0}}\cdot \prod_{i=1}^{\infty}(1-2^{-i})^2 \underset{\autoref{ineq}}{\geq} r_{l_{0}}\cdot1/e^{4}\geq 1/e^{4}\cdot q(n),$$
    where $q(n)$ is any function of $n$ s.t. $r_{l_{0}}\geq q(n)$.
% \frac{1}{25.n^{m^{'}-m+2}.2^n}.

For the rest of the proof, our aim will be to show \autoref{eq:bound_on_r}. We do this by rounding a solution of the linear program for $\mathsf{fmbs}(f^{l+1})$ to a feasible solution for the integer program corresponding to $\mathsf{mbs}(f^{l+1})$. To accomplish this goal, we look at the composed function $f^{l+1}$ as $f\circ f^{l}$, which seems natural given the fact that we have a better understanding of the function $f$.

Let $x:=(x^{1},x^{2},...,x^{n})\in\{0,1\}^{n^{l+1}}$ be an input for which $\mathsf{fmbs}^{z}(f^{l+1})=\mathsf{fmbs}(f^{l+1},x)$ where $x^{1},x^{2},...,x^{n}\in\{0,1\}^{n^{l}}$ and let $y\in\zone^{n}$ be the $n$-bit string corresponding to $x$ i.e.$\; \forall\; i\in[n],\; y_{i}:=f^{l}(x^{i})$. As already mentioned, we will convert the optimal solution $x$ for $\mathsf{fmbs}(f^{l+1})$ to a feasible solution $\hat{x}$ (mentioned later) for $\mathsf{mbs}(f^{l+1})$. 

Let $\{B_{1},...,B_{k}\}$ be the set of all minimal sensitive blocks for $y$. What now needs to be observed is that any minimal monotone sensitive block for $f^{l+1}$ corresponds to a minimal sensitive block $B_{i}$ for $f$. The consequence of this observation is that every feasible solution of $\mathsf{fmbs}(f^{l+1})$ is a feasible solution of the following linear program:
$$\max \sum_{i=1}^{k}w_{i},$$
s.t.
$$\sum_{j:i\in B_{j}}w_{j}\leq \mathsf{fmbs}(f^{l},x^{i}),\; \forall\; i\in[n],$$
$$w_{j}\geq 0, \forall j\in[k].$$

We now convert an optimal solution of $\mathsf{fmbs}(f^{l+1})$ to a feasible solution of $\mathsf{mbs}(f^{l+1})$. Let $\{w^{*}_{j}\}_{j\in[k]}$ be an optimal assignment of weights for the above linear program and let $w^{'}_{j}:=\lfloor w_{j}^{*}\cdot r_{l}\rfloor$. Define $\hat{x}:=(\hat{x}^{1},\hat{x}^{2},...,\hat{x}^{n}),$ where $\mathsf{mbs}(f^{l},\hat{x}^{i}):=\mathsf{mbs}^{y_{i}}(f^{l})$. It can be observed that for all $j\in[k]$ we have:
$$\sum_{j:i\in B_{j}}w^{'}_{j}\leq \sum_{j:i\in B_{j}}w^{*}_{j}\cdot r_{l}\leq \mathsf{fmbs}(f^{l},x^{i})r_{l}\leq\mathsf{fmbs}(f^{l},x^{i})\cdot\frac{\mathsf{mbs}^{y_{i}}(f^{l})}{\mathsf{fmbs}^{y_{i}}(f^{l})} \leq \mathsf{mbs}(f^{l},\hat{x}^{i}),$$
where the last inequality follows from the fact that $f^{l}(x^{i})=f^{l}(\hat{x}^{i})=y_{i}$.

Now consider the following integer program:
$$\max \sum_{i=1}^{k}w_{i},$$
s.t.
$$\sum_{j:i\in B_{j}}w_{j}\leq \mathsf{mbs}(f^{l},\hat{x}^{i}),\; \forall\; i\in[n],$$
$$w_{j}\in \{0,1,2, \dots,\mathsf{mbs}(f^{l})\}, \forall j\in[k].$$

Clearly, $w^{'}$ forms a feasible solution for the above mentioned integer linear program.

We claim that using the assignment $w^{'}_{i}$ defined above we can construct $\sum_{i=1}^{k}w^{'}_{i}$ many disjoint monotone sensitive blocks for $\hat{x}$, which would imply $\sum_{i=1}^{k}w^{'}_{i}\leq \mathsf{mbs}(f^{l+1},\hat{x})$. 

We argue this as follows, consider the minimal monotone sensitive block $B_{1}$ for $y$ and to simplify the discussion assume that $B_{1}:=\{i_{1},\; i_{2}\}$. Now pick the $i_{1}^{th}$ copy of $f^{l}$.
Consider $w^{'}_{1}$ many disjoint monotone blocks for $x^{i_{1}}$ and denote them by $B^{{1}}_{i_{1},1},\; B^{{1}}_{i_{1},2},...\;,B^{1}_{i_{1},w^{'}_{1}}$. Similarly consider $w^{'}_{1}$ many disjoint monotone sensitive blocks for $x^{i_{2}}$. Observe that each of the monotone blocks $B^{1}_{i_{1},1}\cup B^{1}_{i_{2},{1}},\; B^{1}_{i_{1},{2}}\cup B^{1}_{i_{2},2},...,\; B^{1}_{i_{1},w^{'}_{1}}\cup B^{1}_{i_{2},w^{'}_{1}}$ are sensitive for the input $\hat{x}$ and are pairwise disjoint.

This implies:
$$\mathsf{fmbs}^{z}(f^{l+1})r_{l}\leq\sum_{i=1}^{k}w_{i}^{*}r_{l}\leq\sum_{i=1}^{k}(w^{'}_{i}+1)\leq \mathsf{mbs}(f^{l+1},\hat{x})+2^{n}\leq \mathsf{mbs}^{z}(f^{l+1})+2^{n},$$
where the last inequality follows from the fact that $f^{l+1}(\hat{x})=z$ and by the fact that monotone blocks for $y$ are subsets of $[n]$. 

Using the aforementioned inequality, we get:

\begin{align} \label{eq:r_l+1}
r^{z}_{l+1}&=\frac{\mathsf{mbs}^{z}(f^{l+1})}{\mathsf{fmbs}^{z}(f^{l+1})}\geq r_{l}-\frac{2^{n}}{\mathsf{fmbs}^{z}(f^{l+1})} \nonumber \\
&=r_{l}\bigg(1-\frac{2^{n}}{\mathsf{fmbs}^{z}(f^{l+1})}\cdot\frac{\mathsf{fmbs}^{z^{'}}(f^{l})}{\mathsf{mbs}^{z^{'}}(f^{l})}\bigg)\underset{\autoref{mbs and fmbs of composition}}{\geq}r_{l}\bigg(1-\frac{2^{n}}{\mathsf{mbs}^{z^{'}}(f^{l})}\bigg) 
\end{align}
where $z^{'}=\underset{z\in\zone}{\arg\min}\; r^{z}_{l}$.

We fix $l_{0}$ to be the minimum integer s.t. $\mathsf{mbs}(f^{l_{0}})\geq 2.2^{n}$. This gives us:
\begin{align*}
\mathsf{mbs}^{z^{'}}(f^{l})&\underset{\autoref{monotonicity under composition}}{\geq} \mathsf{mbs}^{z^{'}}(f^{l-(l_{0}+1)})\cdot \underset{b\in\zone}{\min}\; {\mathsf{mbs}^{b}(f^{l_{0}+1})}\\
&\underset{\autoref{lemma: fmbs monotone increasing}}{\geq }2^{\lfloor\frac{l-(l_{0}+1)}{2}\rfloor}\cdot(2\cdot 2^{n}).
\end{align*}

Putting the value of $\mathsf{mbs}^{z^{'}}(f^{l})$ in \autoref{eq:r_l+1} gives us 
\autoref{eq:bound_on_r}, completing the proof.
% $$r_{l+1}\geq r_{l}(1-2^{-1-\lfloor\frac{l-(l_{0}+1)}{2}\rfloor}).$$
\end{proof}

\begin{proof}[Proof of \autoref{lift_intro}]

    We will derive the lifting for $\mathsf{fmbs}$ by using a lifting for $\mathsf{fmbs}^{0}$. Formally, if $\mathsf{mbs}(f)=O(M(f)^{\alpha})$, where the complexity measure $M(.)$ composes, then $\mathsf{fmbs}^{0}(f)=M(f)^{\alpha}$.
    
    Let $\mathsf{fmbs}^{0}(f)>M(f)^{\alpha}$ i.e. $\mathsf{fmbs}^{0}(f)=M(f)^{\alpha}+\epsilon$ for some $\epsilon>0$. 
    This implies,
    \begin{align*}
        \mathsf{fmbs}^{0}(f)^{l}=&(M(f)^{\alpha})^{l}(1+\epsilon^{'})^{l}, \intertext{where $\epsilon^{'}:=\epsilon/M(f)^{\alpha}$,}\\
         \mathsf{fmbs}^{0}(f^{l})\geq \mathsf{fmbs}^{0}(f)^{l}&\geq (M(f)^{\alpha})^{l}(1+\epsilon^{'})^{l}.\intertext{ Using \autoref{fmbs composes} and \autoref{thm: ratio}, we get:}\\
         p(n)M(f^{l})^\alpha\geq p(n)\mathsf{mbs}(f^{l})&\geq \mathsf{fmbs}^{0}(f^{l})\geq M(f)^{l\alpha}(1+\epsilon^{'})^{l}.\\\intertext{This implies,}
         p(n)&\geq(1+\epsilon^{'})^{l} 
    \end{align*}
    Which is a contradiction for a fixed $n$ and a sufficiently large $l$.

    Now using the above lifting for $\mathsf{fmbs}^{0}(f)$ we derive the lifting for $\mathsf{fmbs}(f)$. 
    
    Let $\mathsf{fmbs}(f)>M(f)^{2\alpha}$ i.e. $\mathsf{fmbs}(f)=M(f)^{2\alpha}+\epsilon$ for some $\epsilon>0$. 
    This implies,
    \begin{align*}
        \mathsf{fmbs}(f)^{l}=&(M(f)^{2\alpha})^{l}(1+\epsilon^{'})^{l}, \intertext{where $\epsilon^{'}:=\epsilon/M(f)^{2\alpha}$. Now by \autoref{mbs and fmbs of composition} it follows that:}\\
         \mathsf{fmbs}^{0}(f^{2})^{l}\geq \mathsf{fmbs}(f)^{l}&\geq (M(f)^{2\alpha})^{l}(1+\epsilon^{'})^{l}.\intertext{ Using \autoref{fmbs composes} and ~\autoref{thm: ratio}, we get:}\\
         p(n)M(f)^{2l\alpha}\geq p(n)M(f^{2l})^\alpha\geq p(n)\mathsf{mbs}(f^{2l})&\geq \mathsf{fmbs}^{0}(f^{2l})\geq \mathsf{fmbs}(f)^{l}\geq M(f)^{2l\alpha}(1+\epsilon^{'})^{l}.\\\intertext{This implies,}
         p(n)&\geq(1+\epsilon^{'})^{l} 
    \end{align*}
    Which is a contradiction for a fixed $n$ and a sufficiently large $l$.

    Now, if the complexity measure $M$ satisfies the condition $M(1-f)=O(M(f))$ then using the fact that $\mathsf{fmbs}^{1}(f)=\mathsf{fmbs}^{0}(1-f)$, we have:
    $$\mathsf{fmbs}(f)=\max\{\mathsf{fmbs}^{0}(f),\; \mathsf{fmbs}^{1}(f)\}=O(M(f)^{\alpha}). $$
\end{proof}
\section{Results needed for the proof of ~\autoref{thm: ratio}}
\label{appendix:results for theorem1}
We basically give an analog of the identities that hold for standard complexity measures for monotone complexity measures presented in \cite{Tal13}. We start by showing how $\hsc(f)$ and $\mathsf{mbs}(f)$ are related.
\begin{theorem}\label{HSC vs mbs and s}
    For a Boolean function $f:\{0,1\}^{n}\rightarrow\{0,1\}$ we have:
    
    $$\hsc^{z}(f)\leq \mathsf{mbs}^{z}(f)\ms^{1-z}(\tilde{f}),$$
    where $\tilde{f}(x):=f(1-x_{1},1-x_{2},...,1-x_{n})$.
\end{theorem}
\begin{proof}
Let $\hsc(f)=\hsc(f,x)$. We now consider the function $f_{x}$ at the input $0^{n-|x|}$ along with a set of disjoint minimal sensitive blocks $\{B_{1},B_{2},...,B_{k}\}$ with $k=\mathsf{mbs}(f,x)$. We claim that for every $i\in[k]$, $f_{x}$ is sensitive on $0^{\oplus{B_{i}}}$ at each index $j\in B_{i}$. If this was not the case then there would exist $B\subsetneq B_{i}$ s.t. $B$ is a sensitive block for $f_{x}$ at $0^{n-|x|}$, contradicting the claim that $B_{i}$ is minimal. 

Now, we claim that the set $\cup_{i\in[k]}{B_{i}}$ is a certificate for $f_{x}$ at $0^{n-|x|}$. If this was not the case then we would have obtained a sensitive block $B$ at $0^{n-|x|}$ s.t. $B\cap B_{i}=\emptyset$ for all $i\in[k]$. This would have contradicted the assumption that $\{B_{i}:i\in [k]\}$ is a witness for $\mathsf{mbs}(f,x)$. 
Hence from the above discussion we obtain:

$$\hsc(f)=\hsc(f,x)\leq \sum_{i\in[k]}|B_{i}|\leq \mathsf{mbs}(f,x)\ms(\tilde{f})\leq \mathsf{mbs}(f)\ms(\tilde{f}).$$
\end{proof}

In the following theorem we show that the composition result for Boolean functions $f,g$ (see \cite{Tal13}), can also be extended to $\mathsf{mbs}$.

\begin{theorem}
\label{theo: mbs of composition}
Let $f:\zone^n \to \zone$ and $g: \zone^m \to \zone$ be Boolean functions then $\mathsf{mbs}(f \circ g) \leq \fbs(f). \mathsf{mbs}(g)$.
\end{theorem}

\begin{proof}
    Let $f:\{0,1\}^{n}\rightarrow\{0,1\}$ and $g:\{0,1\}^{m}\rightarrow\{0,1\}$ be two Boolean functions. Now consider an input $(x^{1},x^{2},...,x^{n})\in(\{0,1\}^{m})^{n}$ where $x^{i}:=(x^{i}_{1},...,x^{i}_{m})$ for $i\in[n]$.

    Now we wish to calculate $\mathsf{mbs}(f\circ g,x)$. Any minimal monotone block for $f\circ g $ at $x$ is a union of minimal monotone blocks for different $g^{i}$ where $i\in\mathcal{I}\subseteq[n]$ where $\mathcal{I}$ is a minimal sensitive block for $f$ at $y:=(g(x^{1}),..,g(x^{n}))$.

    If we assume that $\mathcal{B}_{1},..,\mathcal{B}_{k}$ is the set of all minimal sensitive blocks for $f$ at $y$ and by $m_{j}$ let us denote the number of minimal monotone sensitive blocks for $f\circ g$ that intersects with $\mathcal{B}_j$. As we want to obtain a collection of disjoint monotone sensitive blocks for $f\circ g$ at $x$ hence we have that $\forall i\in[n],\; \sum_{j:i\in \mathcal{B}_{j}}m_{j}\leq \mathsf{mbs}(g,x^{i}).$  
%\newpage
    Hence $\mathsf{mbs}(f\circ g,x)$ is equal to the optimal value of the following integer program:

    $$\max\sum_{i=1}^{k}m_{i},$$
    s.t.
    $$\sum_{j:i\in \mathcal{B}_{j}}m_{j}\leq \mathsf{mbs}(g,x^{i}),\; \forall i\in[n],$$
    $$m_{i}\in\{0,1,..,\mathsf{mbs}(g)\}$$

Now, relaxing the above integer program to linear program we get a feasible solution for the linear program corresponding to $\fbs(f,y)$ by simply dividing it by $\mathsf{mbs}(g)$.:
% $$\max\sum_{i=1}^{k}m_{i},$$
%     s.t.
%     $$\sum_{j:i\in \mathcal{B}_{j}}m_{j}\leq \mathsf{mbs}(g,x^{i}),\; \forall i\in[n],$$
%     $$m_{i}\geq 0.$$
    This implies,
    $$\fbs(f,y)\geq \frac{OPT(LP(\mathsf{mbs}(f\circ g,x)))}{\mathsf{mbs}(g)}\geq \frac{OPT(IP(\mathsf{mbs}(f\circ g,x)))}{\mathsf{mbs}(g)}=\frac{\mathsf{mbs}(f\circ g,x)}{\mathsf{mbs}(g)} $$

   Taking $x$ to be the input s.t. $\mathsf{mbs}(f\circ g)=\mathsf{mbs}(f\circ g,x)$ then we get our desired result.
\end{proof}

The next lemma shows that how 
% Let us look at the corresponding version of lemma \ref{infinity}.
% For monotone complexity measures part 1 and part 4 of Lemma \ref{infinity} do not hold for non-monotone functions in general i.e. there exist non-monotone functions for which $\mathsf{mbs}(f)=1$, like the ODD-MAX-BIT function which is defined as:
% $$OMB(x):=x_{1}-x_{1}x_{2}+...+(-1)^{n}x_{1}...x_{n}.$$
%  The other parts, i.e. part 2 and 3 hold for the monotone setting as well and have been captured formally by the following lemma: 
% \begin{lemma}
% \label{mbs and fmbs of composition}
%    Let $f, g$ be two Boolean function where $f$ is non-monotone and $z \in \{0,1\}$ then,\\
%            1. $\mathsf{mbs}^z(f \circ g) \geq \mathsf{mbs}(g)$,\\
%            2. $\mathsf{fmbs}^z(f \circ g) \geq \mathsf{fmbs}(g)$,\\
% \end{lemma}

\begin{proof}[Proof of \autoref{mbs and fmbs of composition}]
1. Let $\mathsf{mbs}(g)=\mathsf{mbs}(g,x)$ and assume that $g(x)=0$.
Now as we know that $f$ is non-monotone hence there exists inputs $x^{1},x^{2},x^{3},x^{4}$ s.t. $x^{1}<x^{2}$ and $x^{3}<x^{4}$ and $f(x^{1})\neq f(x^{2})$ and $f(x^{3})\neq f(x^{4})$. In fact we can consider the stronger assumption $|x^{1}-x^{2}|=|x^{3}-x^{4}|=1$ i.e. have hamming distance 1.

For ease of discussion let us assume $f(x^1)=f(x^{4})=0$ and $f(x^{2})=f(x^{3})=1$.

Now consider inputs $y^i\equiv(y^i_{1},...,y^{i}_{n})$, $i\in\{1,2,3,4\}$, for $f\circ g$ which are defined as follows:
$$y^i_{j}:=\begin{cases}
    x & \text{, if $x^{i}_{j}=0$}\\
 \alpha & \text{, o.w.}
\end{cases},$$
where $\alpha$ is any string in $g^{-1}(1)$.

What we claim is that $\mathsf{mbs}(f\circ g ,y^1)\geq \mathsf{mbs}(g)$. This is because we can convert the string $x^{1}$ to $x^{2}$ by flipping the corresponding bits in $y^{1}$.  As $f\circ g(y^1)=f(x^1)=0$ hence we have $\mathsf{mbs}^{0}(f\circ g)\geq \mathsf{mbs}(f\circ g,y^1)\geq \mathsf{mbs}(g).$ Similarly, we can convert the string $x^{3}$ to $x^{4}$ by flipping the corresponding bits in $y^{3}$ to obtain  $\mathsf{mbs}^1(f\circ g)\geq \mathsf{mbs}(f\circ g ,y^3)\geq \mathsf{mbs}(g)$.

If it was the case that $g(x)=1$ then the definition of $y^{i}$ for $i\in\{1,\; 2,\; 3,\; 4\}$ would have been as follows:
$$y^i_{j}:=\begin{cases}
    x & \text{, if $x^{i}_{j}=1$}\\
 \alpha & \text{, o.w.}
\end{cases},$$
where $\alpha$ is any string in $g^{-1}(0)$. Using a similar argument as done for the case when $g(x)=0$ we would have obtained the same inequality.

2. Let $\mathsf{fmbs}(f)=\mathsf{fmbs}(f,x)$ and let $g(x)=0$. Let $ x^{1},x^{2},x^{3},x^{4}\in\zone^{n}$ and $y^{1},y^{2},y^{3},y^{4}\in(\zone^n)^{m}$ be the strings defined in part 1. 

Now consider the linear program for $\mathsf{fmbs}(f\circ g,y^{1})$ i.e.:
$$\mathsf{fmbs}(f\circ g,y^{1}):=\sum_{w\in\mathcal{W}(f\circ g)}b_{w},$$
s.t.
$$\forall (i,j)\in [n]\times [m],\sum_{w\in\mathcal{W}(f\circ g,y^{1}):(i,j)\in w}b_{w}$$
and,
$$\forall w\in\mathcal{W}(f\circ g ,y^{1}),\; b_{w}\in[0,1].$$

Now, as $|x^{1}-x^{2}|=1$ hence let us consider the copy of $g$ in $f\circ g$, call it $j\in[n]$, which corresponds to the bit where $x^{1}$ and $x^{2}$ differ. As $f(x^{1})\neq f(x^{2})$ and $|x^{2}-x^{1}|=1$ hence all the monotone blocks for the $j-th$ copy of $g$ are monotone blocks for $f\circ g$. Now, let $\{\hat{b}_{w^{'}}:w^{'}\in\mathcal{W}(g,x)\}$ be a feasible solution corresponding to the linear program for $\mathsf{fmbs}(g,x)$.

Consider the following assignment to weights $b_{w}$:
$$b_{w}:=
\begin{cases}
    &\hat{b}_{w},\text{ if w$\in \mathcal{W}(g,x)$}\\
    & 0, \text{ otherwise}
\end{cases}
$$

It is easy to verify that that the above assignment for $b_{w}$ forms a feasible solution for the LP corresponding to $\mathsf{fmbs}(f\circ g,y^{1})$. Hence $$\mathsf{fmbs}^{0}(f\circ g)\geq \mathsf{fmbs}(f\circ g,y^{1})\geq \mathsf{fmbs}(g).$$

Similarly, we can say that:
$$\mathsf{fmbs}^{1}(f\circ g)\geq \mathsf{fmbs}(f\circ g,y^{3})\geq \mathsf{fmbs}(g).$$
\end{proof}

%Proceeding further, we look at the following lemma which helps us in proving the divergence of the sequence $\{\mathsf{mbs}(f^{l})\}_{l\geq 1}$ for non monotone functions $f$ with $\mathsf{mbs}(f)\geq 2$.
% \begin{lemma}\label{monotonicity under composition}
%     For Boolean functions $f:\zone^{n}\rightarrow\zone$ and $g:\zone^m\rightarrow\zone$ we have:
% $$\mathsf{mbs}^z(f\circ g)\geq \max\{\mathsf{mbs}^z(f)\mathsf{mbs}^0(g),\bs^z(f)\min\{\mathsf{mbs}^0(g),\mathsf{mbs}^1(g)\}\}.$$
% \end{lemma}
\begin{proof}[Proof of \autoref{monotonicity under composition}]
Let $p^0,p^1$ be the inputs for which $\mathsf{mbs}^z(g)=\mathsf{mbs}(g,p^z),$ for $z\in\zone$. Let $\mathsf{mbs}^z(f)=\mathsf{mbs}(f,y)$.
Now consider the input $x\equiv(x^{1},x^2,...,x^n)$ defined as:
$$x^i:=p^{y_{i}}.$$
The $\mathsf{mbs}^{z}(f\circ g,x)$ is clearly $\geq \mathsf{mbs}^{z}(f)\mathsf{mbs}^0(g)$. This is because for every disjoint monotone sensitive block corresponding to $\mathsf{mbs}^z(f)$ we have $\mathsf{mbs}^0(g)$ many disjoint sensitive blocks.

Similarly, for every disjoint sensitive block corresponding to $\bs^z(f)$, we have $\min\{\mathsf{mbs}^0(g),\mathsf{mbs}^1(g)\}$ many monotone sensitive blocks. This gives $\mathsf{mbs}^z(f\circ g)\geq \bs^z(f) \min\{\mathsf{mbs}^0(g),\mathsf{mbs}^1(g)\}$.
\end{proof}

% \begin{corollary}
% \label{lemma: fmbs monotone increasing}
%     Let $f$ be a non-monotone Boolean function with $z \in \{0,1\}$ then, the sequence $\{\mathsf{mbs}^{z}(f^l)\}_{l \in \N}$ is monotone increasing and if $\mathsf{mbs}(f)\geq 2$ then for every $z \in \{0,1\}$ the sequence $\{\mathsf{mbs}^{z}(f^l)\}_{l \in \N}$ tends to infinity.
% \end{corollary}
\begin{proof}[Proof of \autoref{lemma: fmbs monotone increasing}]
    The monotone increasing part of the lemma is obtained using part 1 of \autoref{mbs and fmbs of composition} as follows:
    \begin{align*}
        \mathsf{mbs}^{z}(f^l)\geq \mathsf{mbs}(f^{l-1})\geq \mathsf{mbs}^{z}(f^{l-1}).
    \end{align*}

Now we show that the sequence diverges if $\mathsf{mbs}(f)\geq 2$. In particular, we show that for all $l\geq 2$ and for all $z\in\{0,1\}$, we have:
    $$\mathsf{mbs}^{z}(f^l)\geq 2^{\lfloor l/2\rfloor}.$$

    We prove it via induction on $l$. The base case is for $l=2$. Using \autoref{monotonicity under composition} and the assumption that $\mathsf{mbs}(f)\geq 2$ we obtain:
    $$\mathsf{mbs}^{z}(f^2)\geq \mathsf{mbs}(f)\geq 2$$
    Now for the inductive step consider $l=k$. 
    Using part 1 of \autoref{mbs and fmbs of composition} and  we get, 
    $$\mathsf{mbs}^{z}(f^k)\geq \mathsf{mbs}^{z}((f^{2})^{\lfloor k/2\rfloor} )\geq 
    \mathsf{mbs}^{z}(f^2)\mathsf{mbs}^{0}(f^{{2\lfloor k/2\rfloor-2}}).$$

    Hence by using the induction hypothesis we get that:
    $$\mathsf{mbs}^{z}(f^k)\geq \mathsf{mbs}^{z}(f^2)\mathsf{mbs}^{0}(f^{{2\lfloor k/2\rfloor-2}})
    \geq 2.2^{\lfloor k/2\rfloor-1}=2^{\lfloor k/2\rfloor} $$
\end{proof}

A natural question to ask at this point is do we have a relation for the ratio between $\hsc(f^{l})$ and $\mathsf{mbs}(f^{l})$ similar to $\mathsf{fmbs}(f^{l})$ and $\mathsf{mbs}(f^{l})$? That might be an interesting problem to look at but what we do have is the following simple corollary which follows from the fact that $\hsc(f)=O(\mathsf{fmbs}(f)\log(\spar(f)))$.
\begin{corollary}\label{hsc vs mbslog}
    For a Boolean function $f:\zone^{n}\rightarrow\zone$ we have that for all $l\geq 1$:
    $$\frac{\hsc(f^{l})}{\mathsf{mbs}(f^{l})\log(\spar(f^{l}))}\leq C(n),$$
    where $C(n)$ is a function independent of $l$.
\end{corollary}
\begin{proof}
    Using the fact that $\hsc(f)=O(\mathsf{fmbs}(f)\log(\spar(f)))$ and by \autoref{thm: ratio} we get:
    $$\frac{\hsc(f^{l})}{\mathsf{fmbs}(f^{l})\log(\spar(f^{l}))}\cdot\frac{\mathsf{fmbs}(f^{l})}{\mathsf{mbs}(f^{l})}\leq \frac{\hsc(f^{l})}{\mathsf{mbs}(f^{l})\log(\spar(f^{l}))}=O(p(n)).$$
\end{proof}

Another interesting observation from \autoref{thm: ratio} is 
the following corollary.
% that for any $c>0$ there are infinitely many (non-monotone) Boolean functions $g$ for which $\mathsf{fmbs}(g)\leq \mathsf{mbs}^{1+c}(g)$.

%Formally,
% \begin{itemize}
\begin{corollary}
    For any Boolean function $f$ and constant $c> 0$, there exists a $l_{0}\in\mathbb{N}$ such that for all $l\geq l_{0}$, $$\mathsf{fmbs}(f^{l})\leq \mathsf{mbs}^{1+c}(f^{l}).$$
\end{corollary}

\begin{proof}
Let us assume that no such $l_{0}$ exists i.e. there exists an infinite sequence of integers, say $\{m_{i}\}_{i\geq 1}$ s.t.:
$$\mathsf{fmbs}(f^{m_{i}})>\mathsf{mbs}(f^{m_{i}})\mathsf{mbs}^{c}(f^{m_{i}}),$$
for all $i\geq 1$.

In other words, this implies that:
$$p(n)\geq \frac{\mathsf{fmbs}(f^{m_{i}})}{\mathsf{mbs}(f^{m_{i}})}>\mathsf{mbs}^{c}(f^{m_{i}}),$$
for all $i\geq 1$.

This is a contradiction to the fact that the sequence $\{\mathsf{mbs}(f^{l})\}$ diverges.

\end{proof}

Finally, we also mention the following inequality which has been used in the proof of \autoref{thm: ratio}:

\begin{proposition}\label{ineq}
    $\prod_{i=1}^{\infty}(1-2^{-i})\geq 1/e^{2}$
\end{proposition}

%\end{Prop}
\begin{proof}
We prove the inequality by applying A.M-G.M. inequality on positive real nos. $\{a_{1},...,a_{N}\}$ where $a_{i}:=\frac{1}{1-2^{-i}}$ followed by taking the limit $N\rightarrow \infty$.

Applying A.M.-G.M. inequality on $\{a_{1},...,a_{N}\}$, we get:

$$(a_{1}a_{2}...a_{N})^{1/N}\leq \frac{\sum_{i=1}^{N}a_{i}}{N}$$

This implies,
$$\bigg(\prod_{i=1}^{N}(1-2^{-i})^{-1}\bigg)^{1/N}\leq \frac{\sum_{i=1}^{N}\frac{2^{i}}{2^{i}-1}}{N}.$$

Now simplifying the above inequality we obtain the following set of inequalities:
\begin{align*}
    \bigg(\prod_{i=1}^{N}(1-2^{-i})^{-1}\bigg)^{1/N}&\leq 1+\frac{\sum_{i=1}^{N}2^{-(i-1)}}{N}\\
    &=1+\frac{2}{N}(1-2^{-N})\leq 1+2/N.
\end{align*}
This implies,

$$\prod_{i=1}^{N}(1-2^{-i})^{-1}\leq (1+2/N)^{N}.$$

Taking $N\rightarrow \infty$ we get the desired inequality.
\end{proof}
\subsection{Characterization of Boolean functions with $\mathsf{mbs}(f)=1$}\label{char of mbs=1}
% As already mentioned in the previous section, there exist non-monotone functions $f:\zone^{n}\rightarrow\zone$ with $\mathsf{mbs}(f)=1$ e.g. ODD-MAX-BIT. What we show in this section is that every non-monotone function with $\mathsf{mbs}(f)=1$ can in fact be expressed in terms of ODD-MAX-BIT function of monomials. Before discussing further, we define the ODD-MAX-BIT on monomials. Say $X_{S_{1}},\; X_{S_{2}},\; ...,\; X_{S_{k}}$ are monomials then:
% $$ODD-MAX-BIT(X_{S_{1}},X_{S_{2}},...,X_{S_{k}}):=X_{S_{1}}-X_{S_{1}}X_{S_{2}}+X_{S_{1}}X_{S_{2}}X_{S_{3}}-...+(-1)^{k}\prod_{i=1}^{k}X_{S_{i}},$$
% where the product $X_{S}X_{T}:=X_{S\cap T}X_{S\setminus T\cup T\setminus S}$.

In this section, we provide another noticeable difference in the behaviour of $\bs$ and $\mathsf{mbs}$. We already know from \cite{Tal13} that for all non-monotone functions $\bs(f)\geq2$. Interestingly, the same is not true for $\mathsf{mbs}$ i.e. there are non-monotone Boolean functions for which $\mathsf{mbs}(f)=1$. For example, the function:
$$ODD-MAX-BIT(X_{S_{1}},X_{S_{2}},...,X_{S_{k}}):=X_{S_{1}}-X_{S_{1}}X_{S_{2}}+X_{S_{1}}X_{S_{2}}X_{S_{3}}-...+(-1)^{k}\prod_{i=1}^{k}X_{S_{i}},$$
where the product $X_{S}X_{T}:=X_{S\cap T}X_{S\setminus T\cup T\setminus S}$, has $\mathsf{mbs}(f)=1$.

What we now show is that the $ODD-MAX-BIT$ is in fact the ``only" function with $\mathsf{mbs}(f)=1$. 
To prove the aforementioned result we need the following claim about the structure of Boolean functions with $\mathsf{mbs}(f)\leq 1$.

\begin{claim}\label{mbs(f)=1}
    If $f:\zone^{n}\rightarrow\zone$ is a Boolean function with $\mathsf{mbs}(f)= 1$ and $f(0^{n})=0$ then there exists an input $x\in\zone^{n}$ s.t. for all $y\in\zone^{n}$, if $y\ngeq x$ then $f(y)=0$.
    
\end{claim}
\begin{proof}
Consider a Boolean function $f:\zone^{n}\rightarrow\zone$ with $\mathsf{mbs}(f)= 1$ and $f(0^{n})=0$ and let $k$ be the smallest integer s.t. $f(x)=1$ and $|x|=k$. Now if we assume that there exists a $y\ngeq x$ with $f(y)=1$ then we have $\mathsf{mbs}(f,x\wedge y)\geq 2$ which contradicts the assumption of $\mathsf{mbs}(f)= 1$.
\end{proof}

\begin{lemma}
    If $f:\zone^{n}\rightarrow\zone$ is a Boolean function with $\mathsf{mbs}(f)\leq 1$ then $f$ can be expressed as $ODD-MAX-BIT(X_{S_{1}},X_{S_{2}},...,X_{S_{k}})$ or 1-($ODD-MAX-BIT(X_{S_{1}},X_{S_{2}},...,X_{S_{k}}))$ where $X_{S_{i}}$ are monomials corresponding to set $S_{i}\subseteq[n]$.
\end{lemma}

\begin{proof}
    We prove the lemma by applying induction on the arity of the Boolean function $f$ i.e. $n$. For the base case let $n=1$.

    Now if $f$ is a constant function then $f(x)=X_{\phi}$ or $1-X_{\phi}$ where $X_{\phi}=1$. If it is not then $f(x)=x_{1}$ or $1-x_{1}$. We see that the condition is satisfied for the base case of $n=1$.

    Now for the inductive step assume that $f(0^{n})=0$ with $n=k$ and $\mathsf{mbs}(f)\leq 1$. Using \autoref{mbs(f)=1} we have that there exists a $x\in\zone^{n}$ s.t. $f(x)=1$ and for all $y\ngeq x$ we have $f(y)=0$. This implies,
    $$f(x)=X_{S}g(x),$$
    where $S:=supp(x)$ and $g:\zone^{n-|x|}\rightarrow\zone$ is the restriction of $f$ on the support of $x$ i.e. $g:=f_{x}$. As $g$ is the restriction of $f$ on $x$ hence $\mathsf{mbs}(g)\leq 1$. For the case when $\mathsf{mbs}(g)=0$ i.e. $g$ is constant, we have $f(x)=0=1-X_{\phi}$ or $f(x)=X_{S}$=$OMB(X_{S})$ for the case when $g(x)=0$ or $1$ respectively. This implies we can assume $g$ to be a non-constant function i.e. $\mathsf{mbs}(g)=1$. Now using the induction hypothesis we have:
    $$g(x)=X_{1}-X_{S_{1}}X_{S_{2}}+...+(-1)^{k}X_{S_{1}}X_{S_{2}}...X_{S_{k}}=OMB(X_{S_{1}},X_{S_{2}},...,X_{S_{k}}),$$
    for some $S_{1},...,S_{k}\subseteq[n]$. Note that over here we could have also assumed that $g(x)=1-OMB(X_{S_{1}},...,X_{S_{k}})$

    Hence $f(x)=X_{S\cup S_{1}}-X_{S\cup S_{1}}X_{S_{2}}+...+(-1)^{k}X_{S\cup S_{1}}X_{S_{2}}...X_{S_{k}}=OMB(X_{S\cup S_{1}},X_{S_{2}},...,X_{S_{k}})$ . 
\end{proof}

\section{Relation between $\mathsf{mbs}(f)$ and $\log(\spar(f))$}\label{mbs vs logspar}
%%%%%%%%%%%%%%%%%%%%%%%%%%%%%%%%%%%%%%%%%%%%%%%%%%%%%%%%%%%%%%%%%%

Sensitivity $\s(f)$ and Fourier degree $\deg(f)$ are two very well studied complexity measures on Boolean functions. Huang, in his landmark result~\cite{Huang}, explicitly proved $\deg (f) \leq s(f)^2$ to show sensitivity is polynomially related to other complexity measures. In the other direction, Nisan and Szegedy~\cite{NS94} showed that $\s(f) \leq \deg(f)^2$ around thirty years ago (we still don't know if this relation is tight). The article~\cite{KLMY21} used this relation to show $\mathsf{mbs}(f)=O(\log^{2}(\spar(f)))$. We show that improving upper bound $\mathsf{mbs}(f)=O(\log^{2}(\spar(f)))$ is indeed equivalent to improving the upper bound on degree in terms of sensitivity (a long standing open question).

% The relation $\mathsf{mbs}(f)=O(\log^{2}(\spar(f)))$ is proved in \cite{KLMY21} by using the relation $\s(f) = O(\deg(f)^2)$ proved in~\cite{NS94}. In fact, the proof  shows that if for a constant $\alpha$ $\forall f: \s(f) = O(\deg(f)^\alpha)$, then $\forall f: \mathsf{mbs}(f)=O(\log^{\alpha}(\spar(f)))$.
%  Interestingly, %it happens that improving this aforementioned relationship between $\mathsf{mbs}(f)$ and $\log(\spar(f))$ is in turn improving the relationship between $\s(f)$ and $\deg(f)$.
%  the converse also holds as we show below.

Recall that \autoref{s(f)-deg} states: 
suppose, there exists a constant $\alpha$ such that for every Boolean function $f:\zone^{n}\rightarrow\zone$, $\mathsf{mbs}(f)=O(\log^{\alpha}{\spar(f)})$. Then for every Boolean function $f:\zone^{n}\rightarrow\zone$, $\s(f)=O(\deg^{\alpha}(f))$.

To prove \autoref{s(f)-deg}, we will need the following known relation between Fourier degree and Fourier sparsity (see e.g. the proof of Fact 5.1 in \cite{ODW+}). A proof is included for completeness.

\begin{claim}
For $f:\{-1,1\}^{n}\rightarrow \{-1,1\}$, $\spar(f)\leq 4^{\deg(f)}$.
\end{claim}
\begin{proof}
Let $f: \{-1,1\}^{n}\rightarrow \{-1,1\}$. Define $g: \zone^n \rightarrow \{-1,1\}$ by $g(x_1, x_2, \dots, x_n) = f(1-2x_1, 1-2x_2, \dots, 1- 2 x_n)$ (notice that $\deg(g) = \deg(f)$). 

Let $g(x)=\sum_{S\subseteq[n]}\alpha_{S} \prod_{i \in S} x_i$ be its polynomial representation. Since $g$ is integer-valued, all $\alpha_S$'s are integers \footnote{This can be seen by induction, using the fact that $\alpha_S$ is an integer linear combination of $f(S)$ (where we interpret $S$ as its indicator vector) and $\alpha_T$ for $T \subsetneq S$.}. 
    
Using the polynomial representation of $g$,
    $$f(y) = g(\frac{1-y_{1}}{2},...,\frac{1-y_{n}}{2})=\sum_{S\subseteq[n]}\alpha_{S}\prod_{i\in S}\bigg(\frac{1-y_{i}}{2}\bigg).$$

From this representation of $f$, every Fourier coefficient of $f$ is an integer multiple of $1/2^{\deg(f)}$. Say $\hat{f}(S)=\beta_{S}/2^{\deg(f)}$ for some $\beta_{S}\in\mathbb{Z}$. Using Parseval, $\sum_{S\subseteq [n]} \beta_S^2 = 4^{\deg(f)}$. Since $\beta_S$'s are integers, this implies that sparsity is at most $4^{\deg(f)}$. 
    % Let the Fourier representation of $f$ be $ f(y) = \sum_{S\subseteq[n]} \hat{f}(S) \chi_{S}(y)$.
    % By the fact that the polynomial representation of a function is unique and that $\alpha_{S}\in\mathbb{Z}$ we have:
    % $$\hat{f}(S)=\beta_{S}/2^{\deg(f)},$$
    % for some $\beta_{S}\in\mathbb{Z}$. Here we have also used the fact that $\deg(g)=\deg(f)$.
%
    % Now using $\beta_{S}\in\mathbb{Z}$ and Parseval's theorem, we have
 % 
    % $$\spar(f)=\sum_{S\subseteq[n]:\hat{f}(S)\neq 0}1\leq\sum_{S\subseteq[n]}\beta^{2}_{S} = 4^{\deg(f)}. $$ \qedhere
\end{proof}

\begin{corollary} \label{cor:spar-deg}
    For $g:\zone^{n}\rightarrow\zone$, 
    $\spar(g) \leq 8^{\deg(g)}$.
\end{corollary}
\begin{proof}
%As we already saw in the proof of the previous theorem that $g(x_{1},...,x_{n})=f(1-2x_{1},...,1-2x_{n})$, hence:
Let $f: \{-1,1\}^n \rightarrow \{-1,1\}$ be defined by $f(x)~=~1-2 g(\frac{1-x_{1}}{2},...,\frac{1-x_{n}}{2})$ similar to what was done in the previous claim. This is equivalent to \\
$g(x_{1},...,x_{n})=(1-f(1-2x_{1},...,1-2x_{n}))/2$. By the above claim,
\[
\spar(g)\leq 2^{\deg(f)}\spar(f)\leq 8^{\deg(g)}.
\] 
\qed
\end{proof}    

Now we can prove \autoref{s(f)-deg}.
\begin{proof}[Proof of \autoref{s(f)-deg}]
    Let $w \in \{0,1\}^n$ be such that $\s(f)=\s(f,w)$. Consider the function $\tilde{f}$ defined by
    $$\tilde{f}(x):=f(x \oplus w),$$
    where $x \oplus w$ denotes the bitwise XOR of $x$ and $w$.
    % $$y_{i}:=\begin{cases}
    %     x_{i},& \text{if $w_{i}=0$} \\
    %     1-x_{i},& \text{o.w.}
    % \end{cases}$$

    Observe that $\s(\tilde{f},0^n)=\s(f,w)$. Also, $\deg(\tilde{f}) \leq \deg(f)$ since performing an affine substitution cannot increase the degree.

    Using the given condition, $\mathsf{mbs}(\tilde{f})=O(\log^{\alpha}{\spar(\tilde{f})})$ and the fact $\mathsf{mbs}(\tilde{f},0)=\bs(\tilde{f},0)\geq \s(\tilde{f},0)=\s(f, w)$ it follows that  
    $\s(f) = O(\log^{\alpha}{\spar(\tilde{f})})$.
    Finally, by \autoref{cor:spar-deg} and $\deg(\tilde{f}) \leq \deg(f)$, we get $\s(f) = O(\deg^\alpha(\tilde{f})) = O(\deg^{\alpha}(f))$ as desired.    
\end{proof}

Hence improving the bound on $\mathsf{mbs}(f)$ in terms of $\log(\spar(f))$ is equivalent to improving the upper bound on $\s(f)$ in terms of $\deg(f)$. The other possible approach of improving the bound on $D^{0-dt}(f)$ is to improve the upper bound on $\mathsf{fmbs}(f)$ in terms of $\log(\spar(f))$. 

It turns out that for the class of symmetric and monotone Boolean functions $\mathsf{mbs}(f)= \mathsf{fmbs}(f)$ = $\hsc(f)$, giving a much better upper bound on $\mathsf{fmbs}(f)$ in terms of sparsity of $f$.

\section{Boolean functions with $\ms(f)=\hsc(f)$}\label{mbs-hsc}
%%%%%%%%%%%%%%%%%%%%%%%%%%%%%%%%%%%%%%%%%%%%%%%%%%%%%%%%%%%%%%%%%%

% \begin{theorem}[\cite{Tal13}]\label{bs less than fbs}
% Let $f:\zone^n \to \zone$ be a Boolean function and $x \in \zone^n$ then 
% $$ \bs(f,x) \leq \fbs(f,x) = FC(f,x) \leq C(f,x)$$.
% \end{theorem}

% \begin{theorem}
% \label{thm: mbs less than fmbs less than hsc}
% Let $f:\zone^n \to \zone$ be a Boolean function and $x \in \zone^n$ then 
% $$ \mathsf{mbs}(f,x) \leq \mathsf{fmbs}(f,x) =FHSC(f,x)\leq \hsc(f,x)$$.
% \end{theorem}

% \begin{proof}
% For a function $f:\zone^{n}\rightarrow\zone$ we already know:
% $$\alpha(f,x)=\alpha(f_{x},0^{n-|x|}),$$
% where $\alpha\in\{bs,fbs,C\}$.

% Hence by using Theorem \ref{bs less than fbs} we obtain:
% $$bs(f_{x},0)\leq fbs(f_{x},0)=FC(f_{x},0)\leq C(f_{x},0)$$
% or,
% $$ \mathsf{mbs}(f,x) \leq \mathsf{fmbs}(f,x) =FHSC(f,x)\leq \hsc(f,x)$$.

% \end{proof}

% \begin{lemma}[\cite{NS94}]
% Let $f:\zone^n \to \zone$ be a monotone function then $s(f) = bs(f) = C(f)$.
% \end{lemma}

In this section, we look at classes of Boolean functions for which $\mathsf{mbs}(f)=\hsc(f)$. We get that for such class of functions $\hsc(f)=O(\log^{2}(\spar(f)))$ (an improvement over the relationship $\hsc(f)=O(\log^{5}{\spar(f)})$ proved in \cite{KLMY21}).

% In the following theorem we show that the above relation holds for monotone and symmetric Boolean functions. 

\autoref{thm:symm_monotone} states that if $f:\zone^{n}\rightarrow \zone$ is monotone or symmetric Boolean function, then 
$$\ms(f) = \mathsf{mbs}(f) = \mathsf{fmbs}(f) = \hsc(f).$$  
% \begin{enumerate}
% \item $f$ is monotone
% \item $f$ is symmetric
% \end{enumerate} 
\begin{proof}[Proof of \autoref{thm:symm_monotone}]

% \begin{lemma}
% \label{lemma: mbs vs. hsc}
% Let $f:\zone^n \to \zone$ be a monotone function then $\ms(f) = \mathsf{mbs}(f) = \hsc(f)$.
% \end{lemma}
%\begin{proof}
We prove the two cases separately.

Case 1 ($f$ is monotone): it suffices to show that for monotone Boolean functions $\hsc(f)\leq \ms(f)$.
Let $x$ be an input for which $\hsc(f,x)=\hsc(f)$.

Without loss of generality assume that $f$ is monotonically increasing and $f(x)=0$. Let $C$ be a minimal certificate for $f_{x}$ at the corresponding all $0^{n-|x|}$ string s.t. $|C|=\hsc(f,x)$. Now consider the input $y$ defined as follows:
$$y_{i}:=\begin{cases}
1 & i\text{$\notin$ C}\\
0 & i\text{$\in$ C}
\end{cases}$$

As $y$ agrees with the all zero string $0^{n-|x|}$ at the indices in $C$ we have that $f_{x}(0)=f_{x}(y).$
Now we claim that each of the indices in $C$ is sensitive for $f_{x}$ at $y$. If it wasn't the case then there exists an $i\in C$ s.t. $f_{x}(y^{\oplus i})=f_{x}(y)$. As $f_{x}$ is also monotone hence for all $z\in\zone^{n-|x|}$ for which $\; y^{\oplus i}\geq z$ we have $f_{x}(y^{\oplus i})\geq f_{x}(z)$.

For any string $z^{'}\in\zone^{n-|x|}$ that agrees with $0^{n-|x|}$ at the bits in $C\setminus{i}$ notice that $z^{'}\leq y^{\oplus i}$. Hence $0=f_{x}(y^{\oplus i})\geq f_{x}(z^{'})=0$. But this implies that $C\setminus \{i\}$ is a certificate for $f_{x}$ at $0^{n-|x|}$ which is a contradiction.

Hence, $$\hsc(f)=\hsc(f,x)=\ms(f_{x},y)\leq \ms(f_{x})\leq \ms(f).$$
% \end{proof}

% % \begin{corollary}
% %    If $f:\zone^{n}\rightarrow\zone$ is monotone then $\hsc(f)=fmbs(f)=mbs(f)=O(log^{2}spar(f)).$
% % \end{corollary}

% Interestingly, lemma \ref{lemma: mbs vs. hsc} also holds for symmetric functions.

% \begin{lemma}\label{symmetric}
%     For a symmetric function $f:\zone^{n}\rightarrow\zone$:
%     $$\ms(f)=\mathsf{mbs}(f)=\hsc(f).$$.
% \end{lemma}
%     \begin{proof}

Case 2 ($f$ is symmetric): again, it suffices to show that if $f:\zone^{n}\rightarrow\zone$ is symmetric then
       $$\hsc(f)\leq \ms(f).$$

       % Let $\hsc(f,x)=\hsc(f)$ with $|x|=k$ and wlog assume that $f(x)=0$. Now if we assume that $f(y)=0$ for $|y|=k+1$ then for every $y\in\zone^{n}$ with $|y|=k+1$ and $y\geq x$ we have that $C(f_{x},0)\leq n-k-1$. This gives us:
       % $$\hsc(f,x)=\hsc(f,y).$$

       % Hence we can assume that $f(y)=1$ for $|y|=k+1$, implying:
       % $$\ms(f)\geq \ms(f,x)=n-|x|\geq \hsc(f,x)=\hsc(f).$$

       Let $x\in\zone^{n}$ be one of the inputs for which $\hsc(f)=\hsc(f,x)=C(f_{x},0^{n-|x|})$. Now, let $C$ be the witness for $\hsc(f,x)$. If $|C|=n-|x|$ then this would imply that all the bits of $0^{n-|x|}$ are sensitive for $f_{x}$. For the other case, i.e. $|C|<n-|x|$,what we claim is that $\hsc(f,x)=\hsc(f,z )$ where $z$ is the input s.t.$\supp(z)= \supp(x)\cup\{i:i\in[n]\setminus(\supp(x)\cup C)\}$ i.e. z is set to $1$ at the bits lying in $\supp(x)$ and all the bits not lying in $C$.

       We claim that $\hsc(f,z)=C(f_{z},0^{n-|z|})=|C|=n-|z|$ i.e. all the bits of $0^{n-z}$ are sensitive for $f_{z}$. If we assume this is not the case i.e. $C^{'}\subsetneq C$ is a certificate for $0^{n-|z|}$ with $|C^{'}|=|C|-1$ then it would imply that $C^{'}$ is also a certificate for $0^{n-|x|}$. 
       
       To argue this, say $C^{'}=C\setminus i$. Now for any input $x<y< z$ that agrees with $0^{n-|x|}$ at the bits of $C^{'}$ we have that $f_{x}(0^{n-|x|})=f_{x}(y)$. This is from using the fact that $f_{x}$ is symmetric, hence we can always swap the $i-th$ bit of $y$ with a zero bit in $y$ not lying in $C$. 

       Hence by the above argument we would have that $C^{'}$ is a certificate for $0^{n-|x|}$ as well. But this contradicts the condition that $\hsc(f,x)=|C|$
    %\end{proof}
\end{proof}

% In article~\cite{KLMY21}, they proved $$D^{0-dt}(f) = O(\mathsf{fmbs}(f)\log{(\spar(f))}.$$ 

Using ~\autoref{thm:symm_monotone}, we get the following corollary.
% mentioned in \cite{logrank} along with lemmas \ref{lemma: mbs vs. hsc} and \ref{symmetric} we get the following interesting corollary.
\begin{corollary}
  Consider a Boolean function $f:\zone^n\rightarrow\zone$. %If $f$ is monotone or symmetric then: 
  \begin{enumerate}
  \item If $f$ is monotone, $\hsc(f)=\mathsf{fmbs}(f)=\mathsf{mbs}(f)=O(\log^{2}(\spar(f)))$.
  \item If $f$ is symmetric, $\hsc(f)=\mathsf{fmbs}(f)=\mathsf{mbs}(f)=(1 + o(1))\log (\spar(f))$.
  \end{enumerate}
\end{corollary}
\begin{proof}
    For the statement about monotone functions, we combine \autoref{thm:symm_monotone} with the relation $\mathsf{mbs}(f) = O(\log^2(\spar(f)))$ ~\cite{KLMY21,buhrman01}.

    For symmetric functions, we use the relations $D^{cc}(f \circ \wedge_2) \leq (1 + o(1))\log(\spar(f))$ \cite{buhrman01} and $\mathsf{mbs}(f) \leq D^{cc}(f \circ \wedge_2)$ ~\cite{KLMY21}.
\end{proof}

Note that the bound above for symmetric functions is tight as can be seen by considering the $\OR$ function which has sparsity $2^n - 1$ and monotone sensitivity $n$.

%The above corollary is interesting for it gives an improvement over the bounds:
% $$\hsc(f)=O(\log^{5}{\spar(f)})$$
% and 
% $$D^{cc}(f\circ\wedge_{2})=O(\log^{5}{\spar(f)}\log(n)),$$
% mentioned in \cite{KLMY21} for general Boolean functions $f:\zone^{n}\rightarrow\zone$.

%In the current section we saw the behaviour of monotone complexity measures when $f$ was symmetric or monotone. In the next section we look at the relationship between these monotone complexity measures for a general Boolean function. 

\section{Examples of separations for monotone measures}
\label{appendix:table}
\begin{table}[ht]
    \centering
    \begin{tabular}{c||c|c|c|c|c}
      & $\mathsf{ms}$ & $\mathsf{mbs}$ & $\mathsf{fmbs}$ & $\mathsf{HSC}$ & $\log(\spar)$\\
      \hline
      \hline
       $\mathsf{ms}$ &\cellcolor{black} &$1$ & $1$& $1$ & $2$\\
      % &\cellcolor{black} & & & &\\
      \hline
       $\mathsf{mbs}$ & ?&\cellcolor{black}& $1$&$1$& $2$\\
      % & &\cellcolor{black} & & &\\
      \hline
       $\mathsf{fmbs}$ & ?&$2$&\cellcolor{black} & $1$ & $4$\\
       %& & \cite{GSS16}& \cellcolor{black}& &\\
      \hline
       $\mathsf{HSC}$ &? &?&? &\cellcolor{black} & $5$\\
      % & & & &\cellcolor{black} & \\
      \hline
      \hline
    \end{tabular}
    \caption{Known relations for monotone measures:  Entry $b$ in row $A$ and column $B$ represents, 
for any function $f$,
% and the complexity measure $R$ and $C$, 
$A(f) = O(B(f))^{b + o(1)}$, }
    \label{tab:my_label}
\end{table}
Here, we note some observations about monotone measures that follow from previous work. We have listed all the known relationships between complexity measures in \autoref{tab:my_label}. All non-trivial relationships follow from the results of Knop et. al. \cite{KLMY21}. A similar table for standard measures was compiled in \cite{ABK+}. The relationships between standard measures do not seem to straightforwardly imply relationships between monotone combinatorial measures. On the bright side, almost all the existing separations between classical complexity measures can be lifted for a monotone analogue of complexity measures, essentially in the same way as the proof of \autoref{s(f)-deg} as we explain later. 
For the relation part, it is natural and interesting to ask if monotone measures are polynomially related at all. Note that a row with $\log(\spar)$ doesn't make sense since it is not polynomially related to other monotone complexity measures.

\paragraph{Why we don't need a row for $\log(\spar)$?}

The function $(\AND_n \circ \OR_2)$ (example 2.18 in Knop et al \cite{KLMY21}) has sparsity exponential in $n$ but constant $\hsc$. So $\
log(\spar)$ can not be bounded by any polynomial power of the monotone complexity measures. 

\paragraph{Lifting separations between classical measures for monotone measures:}

  Let $M_1, M_2 \in \{\s, \bs, \fbs, \cert\}$ and let $mM_1, mM_2$ denote their respective monotone analogues. Suppose $f$ achieves a separation $M_1(f) \geq \Omega(M_2(f)^c)$ and suppose $y$ is the input where $M_1(f) = M_1(f, y)$. Consider the shifted function $g$ which maps $x$ to $f(x XOR y)$. Then $mM_1(g, 0^n) = M_1(f, y) = M_1(f) \geq \Omega(M_2(f)^c) = \Omega(M_2(g)^c) \geq \Omega(mM_2(g)^c)$.

We will give a precise example of the fact that classical separations can be lifted easily for monotone measures.

%For example \cite{KT16} have proved that for the classes of non-constant transitive function $\fbs \geq max \{\s(f), \frac{n}{\s(f)}\} \geq \sqrt{n}$. Precisely they have proved for all zero input, $0^n$, which is $\fbs(f, 0^n)) \geq max \{\s(f), \frac{n}{\s(f)}\} \geq \sqrt{n}$ for all non-constant transitive Boolean function $f$. It follows from the definition of monotone measures that $\mathsf{fmbs}(f(0^n)) = \fbs(f(0^n))$, consequently we have the following lemma.

%{\color{blue}
%\begin{lemma}
    %\label{lemma: fbs of monotone function} 
   % Let $f: \zone^n \to \zone$ be a non constant transitive-function then, $\mathsf{fmbs} \geq \max \{\s(f), \frac{n}{\s(f)}\} \geq \sqrt{n}$.
% \end{lemma}
%}
%Similarly, we have talked about different relationship between measures. 
For example in terms of monotone measures \cite{KLMY21} proved that $\mathsf{fmbs}(f) = O(\mathsf{mbs}(f)^2)$ for all Boolean function $f$. It is natural to ask if the relation is tight or not. Consequently it comes to the best known separations between $\fbs$ and $\bs$ and to check if that example works for monotone measures as well. There exist classes of functions given by \cite{GSS16} that gives separations between $\fbs$ and $\bs$. Let us denote the function introduced by \cite{GSS16} by $\mathsf{GSS}$. %then it is known that,

\begin{theorem}[\cite{GSS16}]
\label{th: GSS}
    There exists a family of Boolean functions $\mathsf{GSS}$ for which $\fbs(\mathsf{GSS}) = \Omega(\bs(\mathsf{GSS})^{\frac{3}{2}})$. 
\end{theorem}
Now $\mathsf{GSS}$ function is such that $\fbs(\mathsf{GSS}(0^n))= \Omega( n^{\frac{3}{4}})$ and $\bs(\mathsf{GSS}) = O(n^{\frac{1}{2}})$. Now, from the definition of monotone measures, it follows that $\mathsf{fmbs}(\mathsf{GSS}(0^n))= \Omega( n^{\frac{3}{4}})$ and  $\mathsf{mbs}(\mathsf{GSS}) = O(\bs(\mathsf{GSS}))= O(n^{\frac{1}{2}})$. Consequently, we have the following lemma,

\begin{lemma}
\label{lem: separation between fmbs and msb}
    There exists Boolean function for which, $\mathsf{fmbs}(\mathsf{GSS}) = \Omega(\mathsf{mbs}(\mathsf{GSS})^{\frac{3}{2}})$
\end{lemma}

Note that the above separation is not tight but it matches the best-known separations for standard measures $\fbs$ and $\bs$.

\end{document}